\documentclass[sigconf, nonacm]{acmart}

\newcommand\vldbdoi{XX.XX/XXX.XX}
\newcommand\vldbpages{XXX-XXX}
\newcommand\vldbavailabilityurl{URL_TO_YOUR_ARTIFACTS}
\newcommand\vldbvolume{18}

\newcommand\vldbpagestyle{empty}
\usepackage{subfigure}
\usepackage{multirow}
\usepackage{threeparttable}
\usepackage{balance}
\usepackage[normalem]{ulem}
\usepackage{amsmath,bm}
\usepackage{threeparttable}
\usepackage{color}
\usepackage{url}
\usepackage{appendix}
\usepackage{bbm}
\usepackage{makecell}
\newtheorem{example}{Example}

\newtheorem{myLemma}{Lemma}

\newtheorem{myDef}{Definition}
\usepackage{graphicx}
\usepackage{flushend}
\usepackage{balance}
\usepackage{diagbox}
\usepackage[ruled,vlined]{algorithm2e}
\usepackage{enumitem}
\usepackage{multirow}
\usepackage{booktabs} 
\usepackage{upgreek}
\usepackage{bbding}
\usepackage{pifont}
\usepackage{gensymb}
\usepackage{subfigure}
\usepackage{array}
\usepackage{colortbl}
\usepackage{tikz}

\begin{document}

\title{\textsc{Birdie}:  Natural Language-Driven Table Discovery \\Using Differentiable Search Index}
 
\author{Yuxiang Guo}
\affiliation{%
  \institution{Zhejiang University}
}
\email{guoyx@zju.edu.cn}

\author{Zhonghao Hu}
\affiliation{%
  \institution{Zhejiang University}
}
\email{zhonghao.hu@zju.edu.cn}

\author{Yuren Mao}
\affiliation{%
  \institution{Zhejiang University}
}
\email{yuren.mao@zju.edu.cn}

\author{Baihua Zheng}
\affiliation{%
  \institution{Singapore Management University}
}
\email{bhzheng@smu.edu.sg}

\author{Yunjun Gao}
\affiliation{%
  \institution{Zhejiang University}
}
\email{gaoyj@zju.edu.cn}

\author{Mingwei Zhou}
\affiliation{%
  \institution{Zhejiang Dahua Technology Co., Ltd}
}
\email{zhoumingwei_hz@163.com}

\begin{abstract}
Natural language (NL)-driven table discovery  identifies relevant tables from large table repositories based on NL queries. 
While current deep-learning-based methods using the traditional dense vector search pipeline, i.e., \textit{representation-index-search}, achieve remarkable accuracy, they face several limitations that impede further performance improvements: (i) the errors accumulated during the table representation and indexing phases affect the subsequent search accuracy; and (ii) insufficient query-table interaction hinders effective semantic alignment, impeding accuracy improvements.
In this paper, we propose a novel framework \textsc{Birdie}, using a differentiable search index. It unifies the indexing and search into a single encoder-decoder language model, thus getting rid of error accumulations.
\textsc{Birdie} first assigns each table a prefix-aware identifier and leverages a large language model-based query generator to create synthetic queries for each table. 
It then encodes the mapping between synthetic queries/tables and their corresponding table identifiers into the parameters of an encoder-decoder language model, enabling deep query-table interactions.
During search, the trained model directly generates table identifiers for a given query.
To accommodate the continual indexing of dynamic tables, we introduce an index update strategy via parameter isolation, which mitigates the issue of catastrophic forgetting.
Extensive experiments demonstrate that \textsc{Birdie} outperforms state-of-the-art dense methods by 16.8\% in accuracy, and reduces forgetting by over 90\% compared to other continual learning approaches.

\end{abstract}
\maketitle

\pagestyle{\vldbpagestyle}
\begingroup\small\noindent\raggedright\textbf{PVLDB Reference Format:}\\
Yuxiang Guo, Zhonghao Hu, Yuren Mao, Baihua Zheng, Yunjun Gao, Mingwei Zhou.  \textsc{Birdie}: Natural Language-Driven Table Discovery Using
Differentiable Search Index. PVLDB, \vldbvolume(X): \vldbpages, 2025.\\
\href{https://doi.org/\vldbdoi}{doi:\vldbdoi}
\endgroup
\begingroup 
\renewcommand\thefootnote{}\footnote{\noindent
This work is licensed under the Creative Commons BY-NC-ND 4.0 International License. Visit \url{https://creativecommons.org/licenses/by-nc-nd/4.0/} to view a copy of this license. For any use beyond those covered by this license, obtain permission by emailing \href{mailto:info@vldb.org}{info@vldb.org}. Copyright is held by the owner/author(s). Publication rights licensed to the VLDB Endowment. \\
\raggedright Proceedings of the VLDB Endowment, Vol. \vldbvolume, No. 8\ %
ISSN 2150-8097. \\
\href{https://doi.org/\vldbdoi}{doi:\vldbdoi} \\
}\addtocounter{footnote}{-1}\endgroup

\ifdefempty{\vldbavailabilityurl}{}{
\vspace{0.3cm}
\begingroup\small\noindent\raggedright\textbf{PVLDB Artifact Availability:}\\
The source code, data, and/or other artifacts have been made available at \url{https://github.com/ZJU-DAILY/BIRDIE}.
\endgroup
}

\section{Introduction}
\label{sec:intro}
Tables are a prevalent format for data storage across governmental institutions~\cite{TUS}, businesses~\cite{Auto-BI}, and the Web~\cite{TURL}. They contain vast amounts of information that can drive  decision-making~\cite{ARM-Net}. However, the sheer volume of tabular data complicates the process for users  to locate relevant tables in large repositories or data lakes~\cite{LakeBench}. In response, the data management community has developed table discovery methods that allow users to search for tables based on various query formats, such as keywords~\cite{AdelfioS13,GoogleSearch}, base tables~\cite{Deepjoin,starmine}, and natural language queries~\cite{Solo,OpenDTR}. Natural language (NL) queries, in particular, are user-friendly and empower non-technical users to express their needs more precisely. As an example shown in Figure~\ref{fig:exm1}, assume that a user wants to know who starred in the movie ``on golden pond''. NL-driven table discovery aims to identify $T_1$ from the large table repository, as it contains a cell to answer this query.
Once table $T_1$ is retrieved, tools like NL2SQL~\cite{FinSQL} or Large Language Models~\cite{ReAcTable} can be used to formulate a response to the query.

\begin{figure}[t]
  \centering
  \includegraphics[width=1\linewidth]{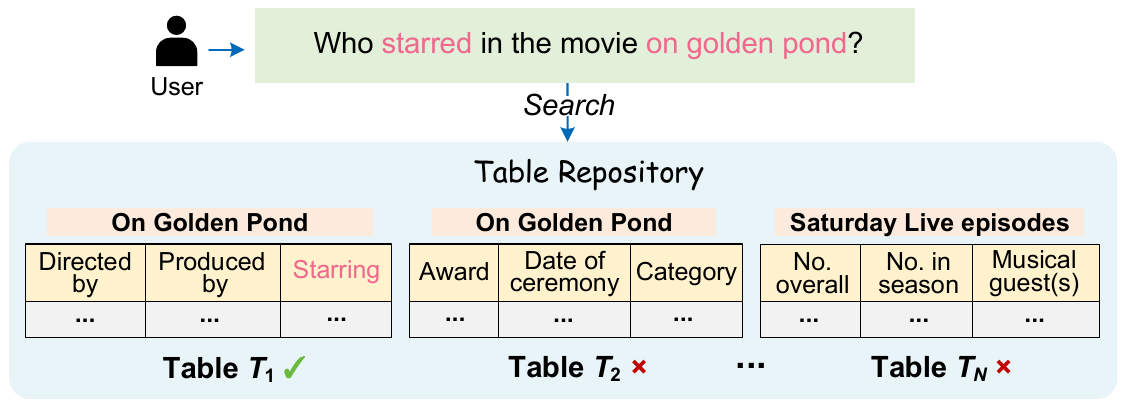}\vspace{-2mm}
  \caption{{An example of NL-driven table discovery.}
  \label{fig:exm1}}
    \vspace{-3mm}
\end{figure}

The success of deep learning techniques in various fields has led to the emergence of the state-of-the-art (SOTA) NL-driven table discovery methods~\cite{Solo, OpenDTR}, which typically follow the traditional dense vector search pipeline involving \textit{representation}, \textit{indexing} and \textit{search}. As illustrated in Figure~\ref{fig:pipeline}(a), a bi-encoder system is trained for representation, comprising a table encoder and a query encoder. The table encoder transforms each table into a fixed-dimensional embedding, and the indexes are constructed on these embeddings. During the search phase, the query embedding is generated using the query encoder, and the nearest neighbor search (NNS) or approximate NNS (ANNS) is conducted to locate table embeddings that closely resemble the query embedding, leveraging pre-constructed indexes. Although this paradigm has significantly outperformed sparse methods (e.g. BM25)~\cite{Solo},  it faces two major limitations.
\begin{itemize} [leftmargin=*]

\item{\textbf{Cumulative Errors from the Multi-Stage Process}.} 
The separation of representation, indexing, and search results in the accumulation of errors from one stage to the next.
First, capturing the rich and complex information within tables using a single vector remains a challenge~\cite{Observatory,TabReps}. Insufficient table representations degrade the effectiveness of search results. Second, the choice of indexing techniques can greatly impact the search results; for instance, discrepancies between the inverted index and ANNS index can negatively affect the accuracy~\cite{starmine}.


\item{\textbf{Insufficient Query-Table Interactions}.}
Encoding the query-table pair using a cross-encoder enables deep query-table interactions,
thus enhancing semantic alignment and improving accuracy~\cite{AdHoc_TR}. However, it is computationally intensive. 
For online search efficiency, current dense search methods typically encode the query and table independently, and rely on basic similarity computations (e.g., cosine similarity) between their embeddings~\cite{Solo,OpenDTR},
which inadequately capture the nuanced query-table interactions necessary for effective retrieval.

\end{itemize}

\begin{figure}[t]
  \centering
  \includegraphics[width=1\linewidth]{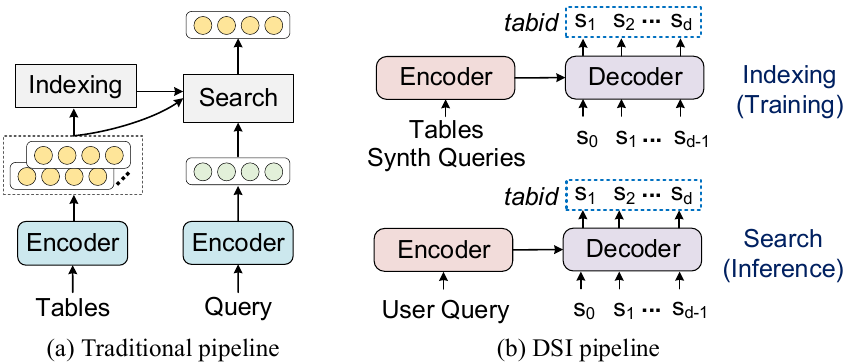}\vspace{-2mm}
  \caption{{ Traditional dense pipeline vs. DSI pipeline.}
  \label{fig:pipeline}}
\end{figure}

To overcome these limitations, this paper introduces \textsc{Birdie}, a novel framework for NL-driven table discovery using a Differentiable Search Index (DSI)~\cite{DSI}. \textsc{Birdie} unifies both indexing and search into 
an encoder-decoder Transformer~\cite{attention} architecture. In this framework, indexing is integrated into the model training process, while search is conducted through model inference, as shown in Figure~\ref{fig:pipeline}(b).
Specifically, each table in the given table repository is assigned a unique table identifier (tabid), represented as a sequence $\mathbf{s}= (s_1, s_2, \dots, s_d)$.
During the training (indexing) phase, the model learns to map each table to its corresponding tabid, and associates generated synthetic queries with the tabids of relevant tables.
This design encodes the table information within the model's parameters, ensuring an end-to-end differentiable training process.
Additionally, training with synthetic queries and tabids fosters deep interactions between queries and tables through encoder-decoder attention~\cite{attention}.
During the inference (search) phase, the trained model receives a query and directly generates tabids token by token.
Despite some successful attempts of DSI in certain applications~\cite{NCI,Tiger,CorpusLM}, three primary challenges emerge in developing an effective table discovery approach based on DSI:

\vspace{0.02in}
\noindent
\textbf{Challenge I: }\textit{How to design prefix-aware tabids to capture the complex semantics in tabular data?} 
The tabid generation during the search phase is autoregressive,
as shown in Figure~\ref{fig:pipeline}(b). Thus, it is crucial to design prefix-aware tabids so that tabids of similar tables share similar prefixes, thereby enhancing accuracy. However, existing ID generation methods~\cite{DSI-QG,Zeng,DSI} are designed for flattened, short documents and are not capable of capturing the complex, hierarchical semantics inherent in tabular data.
To this end, we propose a two-view-based clustering algorithm, i.e., metadata-view and instance-data-view, to generate  tabids for each table.
Through inter-view and intra-view hierarchical modeling, we derive prefix-aware semantic tabids for each table.

\vspace{0.02in}
\noindent
\textbf{Challenge II: }\textit{How to automatically collect NL queries tailored to tabular data for model training?}
Collecting high-quality, diverse, and table-specific queries and associating them with tabids during the training (indexing) phase effectively simulates the subsequent search phase~\cite{DSI-QG}, thus improving search accuracy.
However, collecting queries for large table repositories manually is intractable. Although some synthetic query generation methods have been proposed, they either focus on flattened text paragraphs~\cite{docT5query,NCI} or rely on a two-stage transformation (table-to-SQL and SQL-to-NL)~\cite{Solo}. The former overlooks the structured and non-continual semantics of tabular data, while the latter suffers from quality issues due to the error accumulations across two transformation stages. 
To tackle this, we train a query generator tailored to tabular data using powerful Large Language Models (LLMs),
and design a table sampling strategy to generate diverse and high-quality NL queries. 

\vspace{0.02in}
\noindent
\textbf{Challenge III: }\textit{How to continually index new tables while alleviating the catastrophic forgetting?}
When new tables are added to the repository, they require new tabids and updates to the parameters of the existing model to incorporate new information.
Training the model from scratch using all the tables whenever the repository changes is resource-intensive, while updating model with only new tables can lead to catastrophic forgetting. Although recent studies~\cite{CLEVER, DSI++} suggest replaying some old data during continual training, we still observe several catastrophic forgetting (see experiments in Section~\ref{subsec:index_update}).
To tackle this, we design an incremental algorithm for efficient tabid assignment, and a parameter-isolation method to maintain trained model unchanged while training a memory unit for each new batch of tables.

\vspace{0.02in}
\noindent
\textbf{Solution.} Incorporating techniques that address these challenges, we present \textsc{Birdie}, a novel framework for natural language-driven ta\underline{b}le d\underline{i}scove\underline{r}y via \underline{d}ifferentiable search \underline{i}nd\underline{e}x. The main contributions of this paper are summarized as follows:

\begin{itemize} [leftmargin=*]
\item{} \emph{Differentiable framework.} 
We propose \textsc{Birdie}, an end-to-end differentiable framework for NL-driven table discovery. To the best of our knowledge, it is the first attempt to perform table discovery using differentiable search index,  which completely subverts the convention of previous representation-index-search methods.

\item{} \emph{Prefix-aware tabid construction.} We design a simple yet effective two-view-based clustering approach to model both explicit and implicit hierarchical information embedded in tabular data. Based on this, we assign a prefix-aware tabid to each table, which is well-suited for the autoregressive decoding.

\item{} \emph{LLM-powered query generator.} We continually refine an open-source LLM for better understanding of tabular data, and construct a tailored query generator with a table sampling strategy to create high-quality and diverse NL queries for model training. 

\item{} \emph{Effective continual indexing.} We design an incremental method for tabid assignment that avoids re-clustering, and introduce a parameter isolation-based strategy to effectively index new tables while mitigating catastrophic forgetting.

\item{} \emph{Extensive experiments.}
Our extensive experiments on three benchmark datasets demonstrate the superiority of \textsc{Birdie}, achieving significant accuracy improvements against SOTA methods.
\end{itemize}

The remainder of this paper is organized as follows. Section
~\ref{sec:pre} provides the preliminaries related to our work. Section~\ref{sec:overview} presents the overview of \textsc{Birdie}. 
Section~\ref{sec:ifs} and Section~\ref{sec:index_update} introduce indexing from scratch and index update, respectively.
Section~\ref{sec:exp} reports experimental results and our findings.
Section~\ref{sec:case} provides a case study,
Section~\ref{sec:relatedwork} reviews related works,
and Section~\ref{sec:conlusion} concludes the paper, with directions for future work.

\section{Preliminaries}
\label{sec:pre}
In this section, we first provide the problem statement, and then introduce the backgrounds of large language models and the low-rank adaptation technique.


		 

\begin{figure*}[t]
  \centering
  \includegraphics[width=1\linewidth]{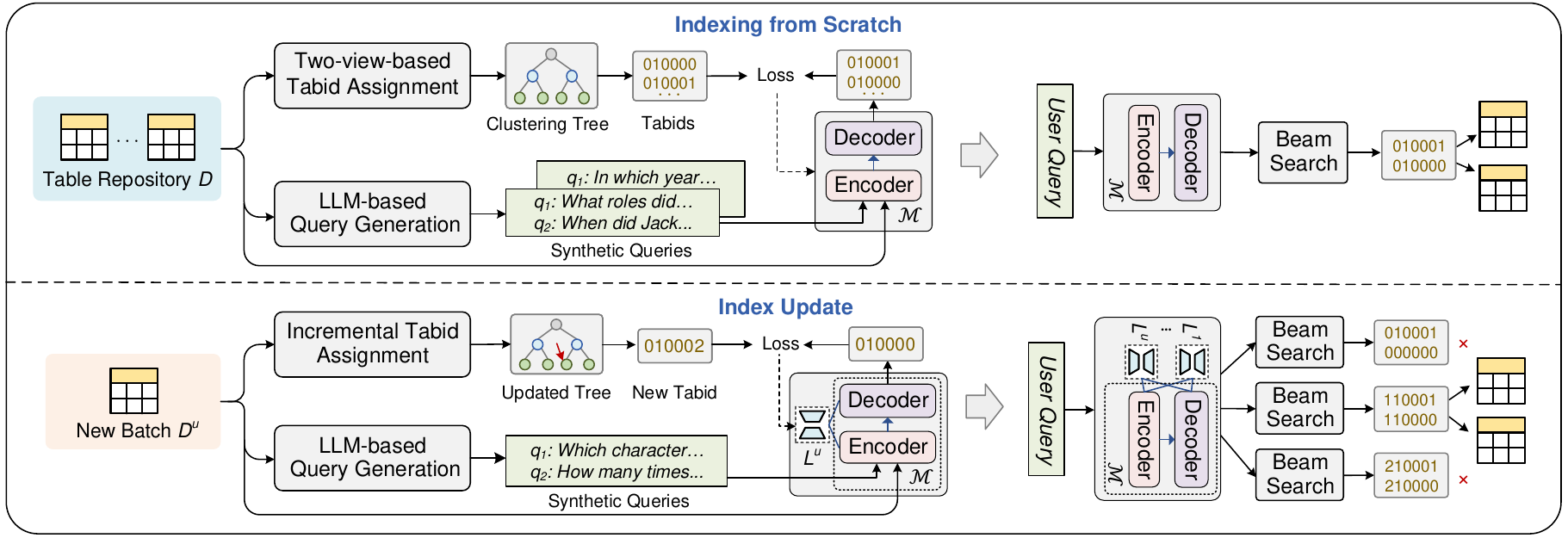}\vspace{-2mm}
  \caption{{The framework of \textsc{Birdie}.
  }
  \label{fig:framework}}
    \vspace{-2mm}
\end{figure*}

\vspace{0.02in}
\noindent
\textbf{Problem Statement.} 
Let $T$ be a relational table with a schema $S = \{A_1, A_2, \cdots, A_n\}$, where each $A_i$ represents an attribute (column). The table may also have a caption (title) $C$, which is a brief description summarizing its content. $T$ consists of $m$ tuples, with each tuple $t_i \in T$ containing $n$ cells, denoted as $e_{ij} = t_i[A_j]$. A table repository $D = \{T_1, T_2, \cdots, T_N \}$ is a collection of such tables.

\begin{myDef}\label{defn:table-discovery}
(NL-Driven Table Discovery). Given a natural language query $q$ and a table repository $D$, NL-Driven Table Discovery aims to search for a table $T^* \in D$ that contains the answer to the query. The answer may be a specific cell within $T^*$ or derived from the reasoning across multiple cells in $T^*$.
\end{myDef}

The objective of NL-driven table discovery is to locate a table relevant to the given query. We leave the scenario where the answer spans multiple tables for future work. Another line of research focuses on ranking a small set of candidate tables based on the query using cross-encoders~\cite{GTR, AdHoc_TR}, which typically serves as a post-retrieval step and can be time-consuming. While Solo~\cite{Solo} proposed a re-ranking model after retrieval, this paper focuses on the retrieval stage, allowing for any existing ranking techniques to be applied after \textsc{Birdie} returns a list of tables. 
 
\begin{myDef}\label{defn:DSI}
(Differentiable Search Index). Given a query space $Q$ and a table repository $D$, a differentiable search index is defined as a function $\mathcal{I}_\theta: Q \rightarrow D$, which is differentiable w.r.t. the parameters $\theta$, allowing optimization via gradient-based methods.
\end{myDef}
 
\vspace{-0.07in}
In traditional search indexes, the mapping $\mathcal{I}(q), q \in Q$, is typically based on predefined, discrete, and non-differentiable operations or structures (e.g., inverted index). In contrast, a Differentiable Search Index parameterizes $\mathcal{I}_\theta$ as a neural model, making the search index differentiable w.r.t. $\theta$ and enabling end-to-end optimization.

\vspace{0.02in}
\noindent
\textbf{Large Language Models.}
Due to the outstanding ability to handle sequential data and capture complex dependencies, Transformer~\cite{attention} has become the main backbone for most large language models (LLMs). 
It consists of an encoder and a decoder, with both components relying on layers of multi-head attention and feed-forward neural (FFN) networks.
Transformer-based LLMs typically adopt one of three architectures: encoder-only, decoder-only, and encoder-decoder. 
Encoder-only models, such as BERT~\cite{BERT}, utilize only the encoder part of the Transformer. These models focus on generating latent embeddings of input text and have found extensive applications in data management~\cite{camper, ZhangMXWX23}.
Decoder-only models, such as GPT-4~\cite{openai_gpt4} and Llama~\cite{llama}, excel in diverse generative tasks, including question answering~\cite{MAR} and NL2SQL~\cite{LLM_NL2SQL}. In this paper, we choose Llama, a widely used open-source decoder-only model, as the base model for query generation.  Finally, encoder-decoder models, also known as sequence-to-sequence models, such as BART~\cite{BART} and T5~\cite{T5}, use an encoder to create a latent representation of the input, which is then passed to the decoder to generate a new sequence. \textsc{Birdie}  uses this encoder-decoder model to construct the differentiable search index, with the encoder processing tables and queries while the decoder generates the corresponding tabids, facilitating a sequence-to-sequence process.

\vspace{0.02in}
\noindent
\textbf{Low-Rank Adaptation}.
As the parameter size of LLMs grows rapidly, full-parameter fine-tuning has become impractical due to the significant resource and time costs involved. To address this, parameter-efficient fine-tuning (PEFT) techniques have been proposed.
Low-Rank Adaptation (LoRA)~\cite{LoRA}, one of the most  recognized PEFT techniques, adjusts the weights of additional rank decomposition matrices and demonstrates comparable effectiveness across various downstream tasks. Specifically, for a pre-trained weight matrix $\mathbf{W}_0 \in \mathbb{R}^{d_1 \times d_2}$, LoRA constrains the update $\Delta \mathbf{W}$ by representing it as a low-rank decomposition $\mathbf{W}_0 + \Delta \mathbf{W} = \mathbf{W}_0 + \mathbf{BA}$, where $\mathbf{B} \in \mathbb{R}^{d_1 \times d_r}$ and $\mathbf{A} \in \mathbb{R}^{d_r \times d_2}$, with $d_r \ll min(d_1, d_2)$. 
Let $\mathbf{h} = \mathbf{W}_0 \mathbf{x}$, then the modified forward pass is: $\mathbf{h}' =  \mathbf{W}_0 \mathbf{x} + \Delta \mathbf{Wx} =  \mathbf{W}_0 \mathbf{x} +  \mathbf{BA}\mathbf{x}$.
%
%
During training, the pre-trained  $\mathbf{W}_0$ remains frozen, with only the low-rank matrices $\mathbf{B}$ and $\mathbf{A}$ being learnable. As $\mathbf{B}$ and $\mathbf{A}$ are much smaller than $\mathbf{W}_0$, the fine-tuning costs for LLMs are significantly reduced. The trained LoRA weights can either be merged with the pre-trained weights or used as a plug-and-play module during inference, without altering the original model.

\section{Overview of \textsc{Birdie}}
\label{sec:overview}

The overview of \textsc{Birdie} is illustrated in Figure~\ref{fig:framework}. \textsc{Birdie} supports both indexing from scratch and updating the index with newly added tables.

Given a table repository $D$, the two-view-based tabid assignment module generates a prefix-aware tabid for each table $T \in D$ by constructing a clustering tree. The tabid $\mathbf{s} = (s_1, s_2, \dots, s_d)$ is a numerical sequence, ensuring that semantically similar tables share similar (or identical) prefixes. 
The tables in the repository are also taken by the LLM-based query generator to produce synthetic queries for each table. Both the generated queries and the serialized tables are input to the encoder-decoder architecture.
The encoder transforms these inputs into latent representations, capturing the essential semantics and information of the sequences. The decoder then uses these latent representations, along with a special beginning-of-sequence token, to generate tabids token by token, until the end-of-sequence token is produced. The model is trained to align the predicted tabids with the true tabids for each table/query using a sequence-to-sequence language modeling loss. After training, we obtain a trained DSI model $\mathcal{M}$. During the search (inference) phase, the trained model takes a user's query as input and generates tabids token by token. To produce a list of probable tabids, \textsc{Birdie} adopts beam search~\cite{BeamSearch}, exploring multiple decoding paths by maintaining a set of the most promising partial sequences at each decoding step.  Ultimately, the corresponding tables are returned to the user.

When a new batch of tables $D^u (u>0)$ arrives, \textsc{Birdie} performs an index update. First, the incremental tabid assignment strategy assigns a tabid for each table $T^u_i \in D^u$,
accounting for the semantics of all previous tables while preserving existing tabids. Next, the query generator generates queries for the new tables. Based on the model $\mathcal{M}$ trained from the original repository $D$, a memory unit $L^u$ is trained on the new tables $D^u$ and the generated queries, using LoRA techniques. During this training, the parameters of $\mathcal{M}$ remain frozen, while only the LoRA weights being trainable. During search, all existing models (including $\mathcal{M}$ and all the memory units) generate tabids in parallel. A query mapping strategy is then employed to select the final tables from the model outputs.

\section{Indexing from Scratch}
\label{sec:ifs}
In this section, we first present the two-view-based tabid assignment mechanism, which assigns a semantic and prefix-aware tabid to each table in the repository. Then, we introduce the LLM-based query generation method to generate synthetic NL queries for DSI training. Finally, we outline the training and inference processes.

\subsection{Two-view-based Tabid Assignment}
\label{subsec:clustering}
Given that a table typically contains many rows and columns, it is impractical for the DSI model to generate the contents of a table directly in response to a given input query. Instead, a practical strategy is to first assign a unique tabid to each table and guide the model to generate the corresponding tabid. The method of tabid assignment is crucial for effective model learning.

Tabids can take the form of numerical sequences or textual descriptions (e.g. a brief abstract of a table). 
While textual tabids align more closely with the Transformer's pre-training process, ensuring their uniqueness can be challenging, especially in large and dynamic table repositories or data lakes. Therefore, we choose numerical tabids to guarantee uniqueness.
A straightforward approach is to assign each table a unique atomic identifier, such as $0, 1, \dots, N-1$. However, this method results in tabids that lack semantic relationships, which limits the Transformer-based model's ability to fully leverage its semantic capture capabilities and complicates the indexing process. 
Additionally, since the tabid generation process of the decoder is autoregressive, each tabid is generated in multiple steps, with each step influenced by previously generated tokens. Considering these factors, we design a two-view-based tabid assignment strategy that ensures (i) tabids encapsulate the semantics of the corresponding tables, and (ii) tables with similar tabid prefixes exhibit higher semantic similarity,  which enhances the autoregressive decoding process.

\vspace{1mm}
\noindent {\textbf{Two Views of Table Semantics.}} We extract two views of each table $T$.
The first view $V_1(T)$ contains high-level semantics about the primary topic of the table. Since the metadata of a table often provides essential context or properties related to its subject or usage, we utilize this metadata as the first view. Specifically, we extract the table's caption $C$ and schema $S = \{A_1, A_2, \dots, A_n\}$, concatenating them to form a sequence: 
$V_1(T) :=   C , A_1, A_2,  \dots ,A_n$.
Note that metadata can sometimes be incomplete in real-world table repositories. To tackle this, one approach could be to train a metadata generation model to augment the tables with additional information. However, in this paper, we will simply skip any missing elements during serializing.

The second view $V_2(T)$ is derived by concatenating the cell values of each table, providing a more fine-grained perspective. To enhance the understanding of different columns, we also include the attribute names in the cell value sequence:
$V_2(T) =  A_1 : e_{11}, \dots A_n : e_{1n} \, \dots, A_1: e_{m1}, \dots, A_n: e_{mn}$.

Subsequently, we employ a pre-trained language model~\cite{Sentence-T5} to encode both $V_1(T)$ and $V_2(T)$, yielding the semantic embeddings $\mathbf{h}^1$ and $\mathbf{h}^2$ for each table.

\begin{algorithm}[!tb]
\DontPrintSemicolon
\small
	\LinesNumbered
    \caption{\small\textbf{Two-view-based Clustering Algorithm (\textsf{TCA})}}
	\label{alg:tca}
	
	\KwIn{the embeddings $\{\mathbf{h}_i^1\}_{i=1}^N$ and $\{\mathbf{h}_i^2\}_{i=1}^N$ of each table $\{T_i\}_{i=1}^N$, the number $k$ of clusters, the  maximum size $c$ of leaf clusters, and the maximum depth $l$ of the first view}
    \label{alg:cluster}
	\KwOut{the root $ \mathcal{C}_R $ of the clustering tree}

    $level \leftarrow 0$,  $\mathcal{C}_R \leftarrow  \{\mathbf{h}_i^1\}_{i=1}^N$\\
    \textsf{HClus} ($\mathcal{C}_R$, $k$, $c$, $level$)\\
 
    \SetKwProg{Fn}{Procedure}{:}{\KwRet}
    \Fn{\textnormal{\textsf{HClus}($\mathcal{C}$, $k$, $c$, $level$)}}{

    \lIf{$level = l$}{
        $\mathcal{C}  \leftarrow \{\mathbf{h}_i^2 | \mathbf{h}_i^1 \in  \mathcal{C}\}$  // view switching
        }
    $\mathcal{C}_{1:k} \leftarrow k$-means$( \mathcal{C}, k)$,  $\mathcal{C}.\textnormal{addChild}(\mathcal{C}_{1:k})$\\
    
        \ForEach{ i $\in \{1, \dots, k\}$}{
            \lIf{$|\mathcal{C}_{i}| > c$}{
                \textsf{HClus}($\mathcal{C}_{i},  k$, $c$, $level + 1$) 
            }
        }
    }
    \Return {$\mathcal{C}_R$}
\end{algorithm}

\vspace{1mm}
    \noindent{\textbf{Clustering Algorithm.}} Building on the two-view embeddings, we design the two-view-based clustering algorithm (\textsf{TCA}). The core idea is to perform clustering based on the semantic embeddings from both views. The first view captures high-level semantics, allowing us to group tables with related topics into the same cluster. We iteratively perform clustering based on $\{\mathbf{h}^1\}$ for $l$ iterations. Then, we switch views, utilizing the fine-grained embeddings $\{\mathbf{h}^2\}$ to further differentiate tables that share similar topics but contain different contents.
The pseudocode of \textsf{TCA} is presented in Algorithm~\ref{alg:tca}.
Given a collection of tables $\{T_i\}_{i=1}^N$ with their semantic embeddings $\{\mathbf{h}_i^1\}_{i=1}^N$ and $\{\mathbf{h}_i^2\}_{i=1}^N$, along with the desired number of clusters $k$, the maximum depth $l$ for the first view, and the maximum size $c$ for leaf clusters, \textsf{TCA} initializes the current level as $0$, sets the root $\mathcal{C}_R$ to $\{\mathbf{h}_i^1\}_{i=1}^N$, and invokes the clustering procedure \textsf{HClus} (lines 1--2). Within \textsf{HClus}, the algorithm performs view switching to use the embeddings $\{\mathbf{h}_i^2\}$ from the second view if the current clustering level reaches $l$ (line 4). Next, tables are grouped into $k$ clusters $\mathcal{C}_{1:k}$ based on their embeddings from the current view,
which are then inserted into the child nodes $\mathcal{C}$ (line 5). For each child $\mathcal{C}_i$, if it contains more than $c$ tables, \textsf{HClus} is recursively applied within that cluster (lines 6--7). 
Finally, \textsf{TCA}  returns the root of the clustering tree (line 8).

This approach enables \textsf{TCA} to construct a clustering tree (as shown in Figure~\ref{fig:tree}) and captures hierarchical relations in two ways: firstly, through view switching, which explicitly incorporates hierarchical structures using metadata and instance data; and secondly, through recursive clustering within each view, which provides an implicit hierarchical structure that refines semantic representations from coarse to fine-grained levels. 

\vspace{1mm}
\noindent{\textbf{Tabid Assignment.}} Based on the two-view clustering tree, we assign a unique tabid to each table.
Specifically, we assign a unique number ranging from $0$ to $k-1$ to each cluster at every level of the clustering tree. For each leaf node containing $c$ tables or fewer, we assign each table within that  cluster a unique number from $0$ to $c-1$. Consequently, each table receives a tabid $\mathbf{s} =( s_1, \dots s_d)$  that represents its unique path from the root to the leaf node in which it resides, where $s_i \in [0, k-1]$ for $i<d$ and  $s_d \in [0, c-1]$.

\vspace{-1mm}
\begin{example}
    Figure~\ref{fig:tree} illustrates an example of a clustering tree where $k=  c = l = 2$. The blue nodes represent clusters obtained using the embeddings $\mathbf{h}_i^1$, while the green nodes represent cluster derived from $\mathbf{h}_i^2$. Each branch is assigned a number within the range $[0, k-1]$, while each table within a leaf node is assigned a number within the range $[0, c-1]$.
Accordingly, $T_2$ is assigned the tabid $01000$, while $T_5$ receives the tabid $01001$.
\end{example}
\vspace{-1mm}

To optimize memory usage, we design  a tree compression (TC) method that uses a compact structure to retain only the essential information for each cluster (node) in the clustering tree. Specifically, each node in the clustering tree maintains the cluster radius $r$,  cluster cohesion $\delta$, cluster center $\mathbf{c}$, and the view $v$ (1 or 2) of the embeddings used to generate this cluster.
The radius $r$ and cohesion $\delta$ of a cluster $\mathcal{C}$ are defined as follows:
\begin{equation}
\small
r(\mathcal{C})=\max\nolimits_{\mathbf{h}_j \in \mathcal{C}} \operatorname{dist}\left(\mathbf{c}, \mathbf{h}_j\right), 
\delta(\mathcal{C}) =\frac{1}{\left|\mathcal{C}\right|} \sum\nolimits_{j=1}^{\left|\mathcal{C}\right|} \operatorname{dist}\left(\mathbf{c}, \mathbf{h}_{j}\right)
\end{equation}
where $\mathbf{c}$ is the center of the cluster, and $\operatorname{dist(\cdot)}$ is a distance function. This approach allows us to manage the tabid assignments dynamically (to be detailed in Section~\ref{subsec:ita}) without the need to store all high-dimensional embeddings $\mathbf{h}^1$ and $\mathbf{h}^2$ for large table repositories.

\begin{figure}[t]
  \centering

  \includegraphics[width=1\linewidth]{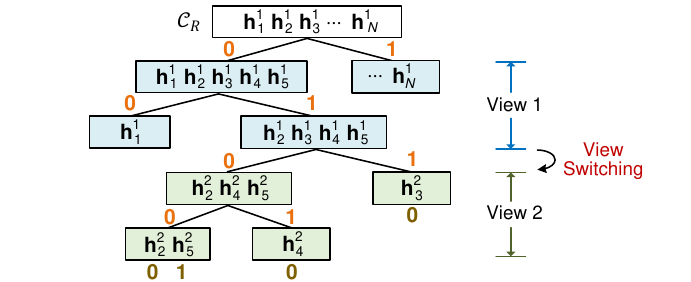}\vspace{-3mm}
  \caption{{ An example of a two-view-based clustering tree. 
  }
  \label{fig:tree}}
    \vspace{-4mm}
\end{figure}

\subsection{LLM-based Query Generation}
\label{subsec:Generator}

Unlike previous query generation methods that either focus on flattened text paragraphs~\cite{docT5query} or rely on manually-designed SQL templates~\cite{Solo},
we leverage the powerful generative capabilities of decoder-only LLMs to directly produce NL queries for tables. While powerful closed-source LLMs like GPT-4~\cite{openai_gpt4} could be used, the costs of API calls and the privacy concerns~\cite{GDPR} often restrict its usage. Consequently, we aim to train a local query generator based on the open-source LLMs.

\begin{table*}[t]
\small
\centering
\caption{Table-related tasks used to continue to train LLaMA3}
\label{tab:table-task}
\vspace{-0.15in}

\renewcommand{\arraystretch}{0.9} 
\setlength{\tabcolsep}{2.8mm}{
\begin{tabular}{|c|c|c|c|c|} 
\hline
    \textbf{Task Name}                   & \textbf{Task Description}                                                                                                   & \textbf{Datasets}                 & \textbf{Difficulty} & \textbf{Size}   \\ 
\hline
Table Size Recognition   & \begin{tabular}[c]{@{}c@{}}Identify the number of rows and columns given a table\end{tabular}                     &       TSR~\cite{MTU}                               & Easy  & 1.5k \\ 
\hline
Table Cell Extraction       & \begin{tabular}[c]{@{}c@{}}Extract the specific cell in a table given the row and column IDs\end{tabular}         &        TCE~\cite{MTU}                              & Easy  & 1.5k \\ 
\hline
Table Row/Column Extraction & \begin{tabular}[c]{@{}c@{}}Extract the specific row (column) in a table given the row (column) ID\end{tabular}   & RCE~\cite{MTU}  & Easy  & 1.5k \\ 
\hline
Table Cell Retrieval         & \begin{tabular}[c]{@{}c@{}}Answer the row and column IDs of a given cell in a table\end{tabular}                  &        TCR~\cite{MTU}                            & Medium & 1.5k \\ 
\hline
Table to Text               & Generate a textual description of a given table                                                                      &         WikiBIO~\cite{Wikibio}                             & Medium & 5k\\ 
\hline
Table Fact Verification     & \begin{tabular}[c]{@{}c@{}}Verify whether a textual hypothesis holds  given tabular data as evidence\end{tabular} &   TABFACT~\cite{Tabfact}                                   & Hard & 5k  \\ 
\hline
Table Question Answering    & Answer the question based on a given a table                                                                                  &       WTQ~\cite{WTQ}                               & Hard  & 5k \\
\hline
\end{tabular} }
\vspace{-2mm}
\end{table*}

\vspace{1mm}
\noindent{\textbf{Training \textsf{TLLaMA3}}}. We adopt the latest LLaMA3 8b~\cite{LLama3} as our base model, as it outperforms other LLMs of similar or larger scales. 
However, since LLMs are mainly pre-trained on textual data, they may struggle to comprehend tabular structures~\cite{TableGPT}. To tackle this, we first fine-tune LLaMA3 on seven fundamental table-related tasks listed in Table~\ref{tab:table-task} that range in difficulty. 
Specifically, we adopt LoRA technique to tune it with a mixture of data derived from these seven tasks. 
We then merge the trained LoRA module with the original pre-trained weights of LLaMA3 to create \textsf{TLLaMA3}.
This process allows \textsf{TLLaMA3} to develop a fundamental understanding of tables, enabling it to perform more complex and specific table tasks through further fine-tuning.

\vspace{1mm}

\noindent{\textbf{Query Generation}}.  Building on \textsf{TLLaMA3}, we construct instruction data to train a query generator tailored to tabular data, as shown in Figure~\ref{fig:instruction}. The instruction requires the model to generate NL queries that can be answered by specific cells in the given table, or derived through reasoning with aggregation operators. It is important to note that some NL queries may lack sufficient specificity for table search tasks. For example, given the table $T_1$ titled ``On Golden Ponds'' (see Figure~\ref{fig:exm1}), the LLM might generate the query ``Who directed this film?''. While this is a common query for a closed-domain QA, it is ambiguous due to the phrase ``this film'', making it less effective for searching tables in repositories. Therefore, we prompt the LLM to  generate more explicit queries by including  necessary table information. The table is transformed into markdown format, and the output consists of a labeled query. 

\begin{figure}[t]
  \centering
  \includegraphics[width=1\linewidth]{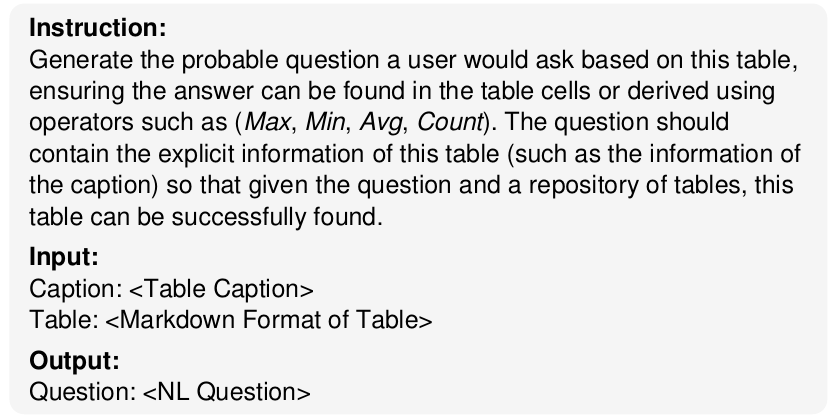}\vspace{-3mm}
  \caption{{Instruction data for query generator training. }
  \label{fig:instruction}}
    \vspace{-5mm}
\end{figure}

Using our constructed instruction data, we adopt LoRA to fine-tune \textsf{TLLaMA3}, resulting in a plug-and-play LoRA module that can be integrated with the weights of \textsf{TLLaMA3} to form our query generator $\mathcal{G}$. 
Subsequently, we use the trained query generator $\mathcal{G}$ to generate multiple NL queries $Q_i$ for each table $T_i \in D$ using the same prompt shown in Figure~\ref{fig:instruction}.  We restrict the maximum input length to 2048 tokens. 
To reduce query generation time,  $\mathcal{G}$ is allowed to generate multiple queries per invocation; however, this can decrease the diversity of the generated queries.
Specifically, we noticed that certain cells in the table were more frequently included in the generated queries, which may be attributed to word frequency bias~\cite{he2019quantifying} during the pre-training phase of LLMs.

To mitigate this issue, we design a table sampling algorithm (\textsf{TSA}) to enrich the generated queries and maximize the coverage of table contents.
The key idea behind \textsf{TSA} is to divide the table into several sub-tables and invoke the generator $\mathcal{G}$ multiple times for each sub-table. This approach allows $\mathcal{G}$ to focus on each sub-table individually.
The pseudo-code for \textsf{TSA} is presented in Algorithm~\ref{alg:tsa}. 
Given a table $T_i$, the number $B$ of required  synthetic queries, and the number $b$ of queries generated per invocation of $\mathcal{G}$, \textsf{TSA} first computes the number $r_s$ of rows for each sampling process, and initializes the row set $R_i$ and query set $Q_i$ (line 1). \textsf{TSA} then iteratively performs table sampling and query generation until $B$ queries are collected.
Specifically, \textsf{TSA} re-initializes $R_i$ to include all rows from $T_i$ if $R_i$ contains fewer than $r_s$ rows (line 3). It next randomly samples $r_s$ rows from $R_i$ and removes the sampled $S(R_i)$ from $R_i$ (lines 4--5) to prevent repeated sampling, which increases the coverage of table information. 
\textsf{TSA} then uses the sampled $S(R_i)$ to invoke the query generator $\mathcal{G}$, generating candidate queries $Q_{cand}$ (line 6).
Given the potential for hallucinations in LLMs, candidate queries may exhibit quality issues. To address this, \textsf{TSA} implements quality checks on $Q_{cand}$ via \textsf{Filter} function that establishes two rules: i) ensure the output format is consistent with ``\texttt{Question: <NL Question>}'' (as shown in Figure~\ref{fig:instruction}), as we found that many low-quality outputs exhibit formatting errors; and ii) check for duplication with existing queries $Q_i$, and perform deduplication if necessary. 
\textsf{TSA} adds the filtered $Q_{cand}$ to the query set $Q_i$ (line 7). Finally, \textsf{TSA} returns the query set $Q_i$ (line 8).

\begin{algorithm}[!tb]
\DontPrintSemicolon
\small
	\LinesNumbered
    \caption{Table Sampling Algorithm}
	\label{alg:tsa}
	
	\KwIn{a table $T_i$,  the number $B$ of required queries, and the number $b$ of generated queries per invocation}

	\KwOut{the generated query set $Q_i$}
    $n_v \leftarrow |B|/|b|$, $r_s \leftarrow |T_i|/n_v$, $R_i \leftarrow T_i$, $Q_i \leftarrow \varnothing$\\
   
    \While{ $|Q_i|$ < B}
    {  
       \lIf{$|R_i| < r_s$}
       {  $R_i \leftarrow T_i$  
       }
       $S(R_i) \leftarrow \textsf{SampleRows}(R_i, r_s)$\\
       $R_i \leftarrow$ remove $S(R_i)$ from $R_i$\\
       $Q_{cand} \leftarrow$ \textsf{Gen}$(\mathcal{G}, S(R_i), b)$\\
       $Q_i \leftarrow Q_i\cup \textsf{Filter}(Q_{cand})$\\
       
    }
    
    \Return {$Q_i$}
\end{algorithm}

\subsection{Model Training and Inference}
After generating the query set $Q_i$ and tabid $\mathbf{s}_i$ for each table $T_i$, we outline the process of constructing differentiable search indexes through model training.
We leverage an encoder-decoder model $\mathcal{M}$ as the backbone. The model is fine-tuned to associate each table $T_i$  with its corresponding tabid $\mathbf{s}_i$, and generate the tabid $\mathbf{s}_i$ based on an input synthetic query derived from table $T_i$.
The model is trained using a standard sequence-to-sequence objective that employs a teacher forcing policy~\cite{Teacher} and minimizes the cross-entropy loss:
\begin{equation}
\footnotesize
\begin{aligned}
\mathcal{L} =\sum\nolimits_{T_i \in D} \big( \log p\left(\mathbf{s}_i \mid \mathcal{M}_\theta\left(T_i\right)\right) +  
\sum\nolimits_{q_{ij} \in Q_i}\log p\left(\mathbf{s}_i \mid \mathcal{M}_\theta (q_{ij})\right) \big)
\label{eq:s2sl}
\end{aligned}
\end{equation}
Here, $Q_i$ represents the set of queries relevant to table $T_i$ generated by our query generator $\mathcal{G}$, and $\theta$ denotes the trainable parameters of model $\mathcal{M}$. For input representation, the table is serialized into a sequence by concatenating its caption, attributes, and cells.

In the search (inference) phase, given an input user query $q$, the fine-tuned model outputs the tabid token by token:
\begin{equation}
\small
p\left(\mathbf{s} \mid q, \theta \right)=\prod\nolimits_{j=1}^d p\left(s_j \mid \mathcal{M}_{\theta}(q, s_0, s_1, \ldots, s_{j-1})\right)
\end{equation}
Here, $\mathbf{s}$ is the output tabid corresponding to the user query $q$, $s_j$ represents the $j$-th token of tabid $\mathbf{s}$, $s_0$ is a  special start token indicating the beginning of the decoding process, and $\theta$ denotes the fine-tuned parameters of model $\mathcal{M}$. 
We employ beam search~\cite{BeamSearch} during the decoding process. This technique explores multiple decoding paths by maintaining a set of the most promising partial sequences at each decoding step, ultimately generating a list of probable tabids. Note that, beam search allows us to attach a probability to each generated tabid, serving as an effective ranking mechanism without incurring the additional computational cost of re-ranking.  With the valid tabid set obtained through two-view-based tabid assignment, we can effectively filter out invalid output tabids during inference.

\section{Index Update}
\label{sec:index_update}

After training on the table repository $D$, we obtain a model $\mathcal{M}$ that encodes all the information of tables in $D$. Since table repositories are typically dynamic~\cite{DataLake_Survey}, it is crucial to continuously index new tables. 
We consider a scenario that $Z$ batches of new tables $\{D^1, D^2, \dots, D^Z\}$ arrive sequentially, with each $D^i$ consisting of newly encountered tables $\{T_1^i, T_2^i, \dots, T_{N_i}^i\}$. The original $D$ is assumed to be significantly larger than the incoming batches, and the new batches generally align with the distribution of $D$. We employ a lazy update strategy, as used in previous studies~\cite{DSI++, CLEVER}, where updates are triggered only when accumulated data exceeds a predefined threshold. 
For clarity, we refer to the original repository as $D^0$ and the model as $\mathcal{M}^0$.

In a dynamic table repository, tables are either inserted or removed. The modifications of tables are naturally handled as deletions followed by insertions.
When a table is removed from the repository, its corresponding tabid is also removed from the valid tabid set. Consequently, any invalid tabids produced by the decoder can be easily filtered out during search. Thus, table deletion is a straightforward operation that does not impact model accuracy, unlike table insertion. Therefore, the remainder of this paper focuses on updating the index for table insertion. In the following, we present the incremental tabid assignment algorithm to attach tabids to new tables without affecting the previous tabids, and illustrate how to index new tables by parameter isolation and how to search for tables under this scenario.

\subsection{Incremental Tabid Assignment}
\label{subsec:ita}
When a new batch $D^u (u>0)$ of tables arrives, the first step is to assign tabids to each new table. \textsf{ReIndex}, a naive method, re-clusters all tables in $\bigcup_{i=0}^{u} D^i$ to account for table semantics and ensure unique tabids using Algorithm~\ref{alg:tca}, and re-assigns tabids. However, \textsf{ReIndex} has two main drawbacks: (i) re-clustering is computationally expensive and time-consuming, and (ii) previously assigned tabids become invalid as clusters change, necessitating a complete retraining of the model from scratch to index all the tables in $\cup_{i=0}^{u} D_i$. To overcome these challenges, we propose an incremental algorithm that assigns tabids to new tables without affecting existing ones. 

We assume that the existing clustering tree is $\mathbb{T}$, and we illustrate the insertion of one new table with embedding $\mathbf{h}_{new}$. First, we find the child node $\mathcal{C}^*$ of the root of $\mathbb{T}$ that is closest to $\mathbf{h}_{new}$:
\begin{equation}
\small
    \mathcal{C}^* = \operatorname{argmin}_{\mathcal{C}} \operatorname{dist}(\mathcal{C}.\mathbf{c}, \mathbf{h}_{new}) 
    \label{eq:closest_cluster}
\end{equation}
where $\mathcal{C}$ represents the child node of the root of $\mathbb{T}$, $\operatorname{dist}(\cdot)$ is a distance metric implemented by Euclidean distance, and $\mathcal{C}.\mathbf{c}$ denotes the center $\mathbf{c}$ of cluster $\mathcal{C}$. We then compute the distance $d$ from $\mathbf{h}_{new}$ to the center $\mathcal{C}^*.\mathbf{c}$. 
For simplicity, we denote the radius, 
cohesion, 
and center of cluster $\mathcal{C}$
as $r$, $\delta$, and $\mathbf{c}$, respectively. We elaborate on the following three cases when inserting $\mathbf{h}_{new}$ into 
$\mathbb{T}$.

\noindent
\textit{\underline{Case I}: Inserting without updates.} If $h_{new}$ is  sufficiently close to cluster $\mathcal{C}^*$, its inclusion has a negligible impact on the properties of $\mathcal{C}^*$.To enhance the efficiency of incremental tabid assignment, updates to the properties can be omitted when $ d \leq  \delta$, i.e., the distance from $h_{new}$ to the cluster center does not exceed the average distance from existing embeddings to the center. In such cases, we directly insert $\mathbf{h}_{new}$ to cluster $\mathcal{C}^*$ without updates.

\noindent 
\textit{\underline{Case II}: Inserting with updates.} When $\delta < d \leq  r$, it suggests that $\mathbf{h}_{new}$ belongs to $\mathcal{C}^*$ but is farther from the cluster center than average.
In this case, we insert $\mathbf{h}_{new}$ into $\mathcal{C}^*$ and update the center $\mathbf{c}$, cohesion $\delta$, and radius $r$ of $\mathcal{C}^*$. We denote the updated values as $\mathbf{c}'$, $\delta'$ and $r'$, with $\mathbf{c}'$ and  $\delta'$ calculated as follows: 
%
\begin{equation}
\small
{\mathbf{c}}^{\prime}=\frac{\left|\mathcal{C}^*\right|  \cdot \mathbf{c} + \mathbf{h}_{new}}{\left|\mathcal{C}^*\right|+1}, \quad\quad {\delta}^{\prime}=\frac{\left|\mathcal{C}^*\right| \cdot \delta + d}{\left|\mathcal{C}^*\right|+1}
\label{eq:c_and_prime}
\end{equation}

\noindent where $|\mathcal{C}^*|$ denotes the size of $\mathcal{C}^*$ after the insertion of  $\mathbf{h}_{new}$.
However,  computing the radius $r'$ by calculating the distance between the new center and all table embeddings is not feasible, as the clustering tree does not retain table embeddings to save storage, as stated in Section~\ref{subsec:clustering}. Therefore, we present the following lemma to estimate the radius $r'$:
%

%
\begin{myLemma}
    The upper and lower bounds of the radius $r'$ are given by $r + \operatorname{dist}(\mathbf{c}, \mathbf{c}^\prime)$ and $\operatorname{dist}(\mathbf{h}_{new}, \mathbf{c}^\prime)$, respectively.
    \label{lemma:r'}
\end{myLemma}
 \vspace{-1mm}

The proof can be found in Appendix~\ref{appendix:A}. Accordingly, we can estimate $r'$ by sampling uniformly between its lower and upper bounds:
\begin{equation}
\small
    r^\prime = \operatorname{Uniform}(\operatorname{dist}(\mathbf{h}_{new}, \mathbf{c}'), r + \operatorname{dist}(\mathbf{c}, \mathbf{c}^\prime))
    \label{eq:r}
\end{equation}

\noindent
\textit{\underline{Case III}: Constructing a new cluster.} The final case occurs when $\mathbf{h}_{new}$ is too distant from the center of the nearest cluster, indicating that  $\mathbf{h}_{new}$ should be assigned to a new cluster. Specifically, if $d > r$, we create a new cluster with $\mathbf{h}_{new}$ as its center, initializing its radius and cohesion to the average values of all child nodes of the root, ensuring consistency with the scale and density of existing clusters at the same hierarchical level. While splitting an existing cluster and inserting the new embedding is a viable alternative, it incurs significantly higher computational costs, making it less preferable. 

After inserting the new embedding $\mathbf{h}_{new}$ into an existing node or a newly created node, the insertion process is recursively applied until a leaf node is reached, at which point a tabid is assigned to the new table. This process is repeated for each new table. The detailed incremental tabid assignment procedure (\textsf{ITA}) is outlined in Algorithm~\ref{alg:ita}.
If the current node is a leaf, \textsf{ITA} generates and returns the tabid (lines 1--2). Otherwise, \textsf{ITA} continues with embedding insertion, determining which embedding view to use (lines 3--4), identifying the closest cluster $\mathcal{C}^*$ using Eq.~\eqref{eq:closest_cluster} (line 5), and calculating the distance $d$ from the new embedding to the center of cluster $\mathcal{C}^*$ (line 6). Based on the value of $d$, one of the following actions is taken: i) if $d$ does not exceed the cohesion of $\mathcal{C}^*$, the embedding is inserted directly into $\mathcal{C}^*$ and the insertion process continues recursively (line 7); ii) if $d$ falls between the cohesion and the radius of $\mathcal{C}^*$, the embedding is inserted into $\mathcal{C}^*$ with cluster information updated, followed by recursive insertion (lines 9--11); and iii) otherwise ($d$ exceeds $r$), a new cluster is created and the embedding is inserted (lines 13--14). 

\begin{algorithm}[!tb]
\DontPrintSemicolon
\small
	\LinesNumbered
	\caption{Incremental Tabid Assignment}
	\label{alg:fmrl}
	
	\KwIn{the current node $\mathcal{C}$, the embeddings $\mathbf{h}^1_{new}$ and $\mathbf{h}^2_{new}$ of a new arrival table $T_{new}$}
	\KwOut{tadid of $T_{new}$}
    \label{alg:ita}
    \If{$\mathcal{C}$ is a leaf node}
    {
    tabid $\leftarrow$ path from the root to $\mathcal{C}$, 
    \Return{\textnormal{tabid}}
    }

    \lIf{$\mathcal{C}.child.v = 1$}
    {$\mathbf{h}_{new} \leftarrow \mathbf{h}^1_{new}$
    }
    \lElse
    {
    $\mathbf{h}_{new}  \leftarrow \mathbf{h}^2_{new}$
    }

    find the child node $\mathcal{C}^*$ of $\mathcal{C}$ that is closest to $\mathbf{h}_{new}$ \/// \text{Eq.~\eqref{eq:closest_cluster}}\\

    $d \leftarrow \operatorname{dist}(\mathcal{C}^*.\mathbf{c}, \mathbf{h}_{new}) $\\
    \lIf{$d \leq \mathcal{C}^*.\delta$}
    {
    insert $\mathbf{h}_{new}$ to $\mathcal{C}^*$,
    $ \textsf{ITA}(\mathcal{C}^*,\mathbf{h}^1_{new}, \mathbf{h}^2_{new})$ 
    }
    \ElseIf{$\mathcal{C}^*.\delta < d \leq \mathcal{C}^*.r$}
    {
    insert $\mathbf{h}_{new}$ to $\mathcal{C}^*$\\
    update the information of  $\mathcal{C}^*$ \/// \text{Eq.\eqref{eq:c_and_prime} and Eq.\eqref{eq:r}}\\
    $ \textsf{ITA}(\mathcal{C}^*,\mathbf{h}^1_{new}, \mathbf{h}^2_{new})$
    }
    \Else{
    create a new child node $\mathcal{C}_{new}$ under $\mathcal{C}$ centered at $\mathbf{h}_{new}$\\
    insert $\mathbf{h}_{new}$ to $\mathcal{C}_{new}$
    }

\end{algorithm}

\subsection{Continual Indexing via Parameter Isolation}
\label{subsec:parameter_isolaiton}
When tabids of new tables in $D^u$ are obtained, a straightforward method to update indexes is to use new tables or combine them with key tables from previous batches to continually train the model~\cite{CLEVER}. However, this strategy can lead to catastrophic forgetting, where the model loses information about older tables that are not included in the continual training process. To tackle this issue, we propose a method that isolates the parameters of the model $\mathcal{M}^0$ trained on $D^0$ from those learned continually. 

\vspace{1mm}
\noindent \textbf{Construction of Memory Units.} For each batch of tables $D^u (u>0)$, we train a LoRA module $L^u$ to index the new tables based on the existing model $\mathcal{M}^0$. The intuition behind this approach is that $\mathcal{M}^0$ has already been trained on $D^0$, which we assume to be a relatively large repository compared to the incoming batches. As a result, $\mathcal{M}^0$ has developed the capability to understand, index, and search for tables effectively. Therefore, we only need to train a LoRA module $L^u$ as a memory unit based on $\mathcal{M}^0$, which is more efficient than performing full fine-tuning. 
Unlike previous studies~\cite{FinSQL,DACE} that typically add LoRA matrices to the weights of self-attention layers in each Transformer block to acquire new capabilities, we argue that this is not optimal for our context. Our goal is to encode information of new tables into the LoRA modules rather than to acquire new abilities. Inspired by the earlier research~\cite{FFN} suggesting that the FFN in Transformer operates as a memory, we add LoRA  modules to FFN in each Transformer block. This allows each batch of new tables $D^u (u>0)$ to have its own dedicated  LoRA memory unit $L^u$.

We store all these units along with the initial model $\mathcal{M}^0$ in a memory hub. These model parameters remain isolated, preventing them from affecting one another. 
After training the LoRA module $L^u$ on the newly arrived tables $D^u (u > 0)$, the model hub contains the following models $\mathcal{M}^i$, where each $\mathcal{M}^i$ indexes tables from a sub-repository $D^i (0 \leq i \leq u)$ and $\bigoplus$ denotes the plug-and-play of $L^i$ during inference.  
$$
    \mathcal{M}^i = \begin{cases} \mathcal{M}^0 & \text{ if } {i = 0} \\
    \mathcal{M}^0\bigoplus{L^i} & \text {if } {1\leq i \leq u} \end{cases}
$$
\noindent\textbf{Search Stage.} During the search phase, each model $\mathcal{M}^i$ takes a user query $q$ as input and performs inference using beam search  either serially or in
parallel, depending on available memory resources. After that, each model obtains the candidate tables $\mathcal{T}^i = \{T^i_1, \dots, T^i_K\} \subset D^i$, with $K$ a user-defined parameter that controls the number of tables returned. However, not all results are equally reliable, as relevant tables for the query $q$ may reside in a few sub-repositories. To effectively select candidate tables from $\{\mathcal{T}^0, \dots, \mathcal{T}^u\}$, we implement a simple yet effective mapping strategy. For each $\mathcal{T}^i$, we first compile the synthetic queries generated during training into a query set $\mathcal{Q}^i = Q^i_1 \cup \dots \cup Q^i_K$, where $Q_j^i$ represents the synthetic queries for table $T_j^i \in D^i$. These query sets collectively form a small query pool $\mathcal{Q}^0 \cup \dots \mathcal{Q}^u$. Next, we encode all queries in the query pool, along with the user query $q$, using a pre-trained semantic-aware encoder~\cite{tsb-roberta-base}. We then identify the top-$n_q$ queries of $q$ in the query pool that are most similar to $q$. Finally, we select the query set $\mathcal{Q}^z$ that includes the largest number of queries found among the top-$n_q$ similar queries, and random selection is used as the tiebreaker. The corresponding candidate tables $\mathcal{T}^z \subset D^z$ are then returned as the results for table discovery. 
%
An extremely small $n_q$ may lead to incorrect selection due to the randomness of the synthetic queries, while a large $n_q$ will consider the contributions of many less relevant queries. Therefore, we set $n_q$ to $\lfloor \frac{B}{4}\rfloor$ in our experiments, where $B$ is the number of synthetic queries generated for each table $T_j^i$, as specified in Algorithm~\ref{alg:tsa}. 

\vspace{1mm}
\noindent \textbf{Discussion.}
Current mainstream libraries (e.g., PEFT~\cite{peft}) 
typically merge each LoRA module individually into the base model before inference, leading to multiple full-model copies being loaded into GPU memory during parallel inference.
One potential optimization is to load a single base model alongside all LoRA modules. During inference, for a given input $\mathbf{x}$, the output of the base model $\mathcal{M}^0(\mathbf{x})$ and the output of each LoRA module $L^i(\mathbf{x}) = \mathbf{B}^i \mathbf{A}^i \mathbf{x}$ are computed in parallel. The output of each LoRA module is subsequently added to $\mathcal{M}^0(\mathbf{x})$ to get $\mathcal{M}^i(\mathbf{x})$. A recent study~\cite{S-LoRA} explores this optimization for parallel inference with multiple LoRAs, utilizing tensor parallelism strategy and highly optimized custom CUDA kernels to efficiently support inference of thousands of LoRA modules on a single GPU.  
However, it is currently limited to decoder-only architectures and does not support our encoder-decoder DSI model. Extending it to encoder-decoder architectures is outside the scope of this paper but remains an avenue for future work.
\section{Experiments}
\label{sec:exp}
In this section, we conduct extensive experiments to demonstrate the efficacy of our proposed \textsc{Birdie}.

\subsection{Experimental Setup}
\noindent \textbf{Datasets.} The experiments were conducted on three real-world benchmark datasets, with statistics  summarized in Table~\ref{tab:dataset}. (i) \textbf{NQ-Tables}~\cite{OpenDTR} is a widely used benchmark for table retrieval.
Each table in NQ-Tables includes a brief caption, though many tables with the same caption differ in schema and content. The dataset contains some inconsistencies, making it relatively noisy. Following the previous work~\cite{Solo}, we use only the $952$ test queries that have a single ground truth table. (ii) \textbf{FetaQA}~\cite{fetaqa} contains clean tables and test NL queries designed for free-form table question answering. Each test query is paired with a unique ground truth table containing the answer.
(iii)  \textbf{OpenWikiTable}~\cite{OpenWiki} is built
upon two closed-domain table QA datasets WikiSQL and WikiTableQuestions. Table descriptions (i.e. page title, section title, caption) are manually re-annotated and then added to the original queries, ensuring this dataset can be applicable to table discovery task.
 

\noindent \textbf{Baselines.} To demonstrate the efficacy of our proposed \textsc{Birdie}, we compare it with the following baseline methods. 1) \textsf{BM25} is a widely used sparse retrieval method. 2) \textsf{DPR} (Dense Passage Retriever~\cite{DPR}) is the SOTA dense retrieval method for open-domain QA of passages. We serialize each table to text and use bert-base-uncased~\cite{BERT} for both query and text encoders, training them with in-batch negatives on the training splits. Mean pooling is adopted to obtain both the query and table embeddings. 3) \textsf{DPR-T} is an extension of \textsf{DPR} to incorporate table structure information. Following~\cite{OpenWiki}, we add special tokens like \text{[Caption]}, [Header], [Rows] during table serialization, using the embedding of the first [CLS] token as the table embedding. 4) \textsf{OpenDTR}~\cite{OpenDTR} is a learning-based table discovery method using TAPAS~\cite{TAPAS} for both query and text encoders. 5) \textsf{Solo-Ret} is the first-stage retrieval module of \textsf{Solo}~\cite{Solo}, which is a SOTA learning-based NL-driven table discovery system comprising two stages: retrieval and ranking.  To ensure a fair comparison, we focus solely on the retrieval module Solo-Ret. 6) \textsf{Solo}~\cite{Solo} reranks results from \textsf{Solo-Ret}. We follow its original implementation to retrieve the top-250 triplets, then apply second-stage ranking.

\begin{table}[t] 
\small
\centering
\caption{Statistics of the datasets used in experiments.}
\renewcommand{\arraystretch}{1} 
\vspace{-0.15in}
\label{tab:dataset}
\setlength{\tabcolsep}{2.5mm}{
\begin{tabular}{l|c|c|c|c}

\specialrule{.08em}{.06em}{.06em}
Dataset   &  $\#$ Tables & Quality &  $\#$ Train &  $\#$ Test \\ 
 \hline
NQ-Tables &  169,898      &    dirty    &   9,534     &    952                  \\
FetaQA  &  10,330     &    clean   &   7,326    &   2,003                   \\
OpenWikiTable   &24,680     &  clean &  53,819  &  6,602                   \\
\specialrule{.08em}{.06em}{.06em}
\end{tabular}}
\vspace{-3mm}
\end{table}



\noindent \textbf{Evaluation metrics.}  To evaluate the effectiveness, we follow the previous study~\cite{Solo} to adopt the precision-at-$K$ ($P@K$) metric, aggregated over the test queries. $P@K$ measures the proportion of test queries for which the ground truth table appears among the top-$K$ results returned by each method. 
Since each query in our benchmark dataset has a single ground truth table, precision alone suffices for evaluating the effectiveness~\cite{Solo}.
In addition, we measure runtime to assess efficiency.

\noindent \textbf{Implementation details.} \textsc{Birdie} is implemented in PyTorch. For tabid assignment, we use the encoder of pre-trained Sentence-T5~\cite{Sentence-T5} to obtain the embeddings, and fit as much of the second view as possible within the default 512-token limits. The number $k$ of clusters and the maximum size $c$ of  leaf clusters are set to 32 for NQ-Tables and 20 for FetaQA and OpenWikiTable, while the depth $l$ of the first view is set to 2 for all three datasets. 
For training the query generator, we fine-tune  \textsf{TLLaMA3} using LoRA on the training split of each dataset, adding LoRA to the attention layers with a rank $d_r$ of 8. 
We generate $B = 20$ queries per table by default. We use mT5-base~\cite{mt5-base} as the backbone for \textsc{Birdie},
and set the maximum length of the input data to 64 tokens.
The batch size is set to 64, and the learning rate is set to 5e-4. For training the memory units, we add LoRA to FFNs of both the encoder and decoder with $d_r = 8$.
During inference (search), we adopt beam search with a beam size of 20. 
All experiments were conducted on a server with an Intel(R) Xeon(R) Silver 4316 CPU (2.30GHz), 6 NVIDIA 4090 GPU (24G each), and 256GB of RAM. 
All programs were implemented in Python.

\subsection{Evaluation of Indexing from Scratch}

Table~\ref{tab:overall} summarizes the $P@1$, $P@5$, and the search time (Time) for \textsc{Birdie} and the baseline methods, averaged across all test queries.

\noindent \textbf{Effectiveness.} Several key observations can be made as follows.
(i) \textsf{BM25} exhibits poor accuracy  on NQ-Tables but performs notably better on FetaQA and OpenWikTable. This disparity arises because the queries tested in NQ-Tables are brief, providing insufficient context for effective sparse retrieval. In contrast, queries in the other two datasets, especially OpenWikiTable, are more detailed, which benefits the performance of sparse methods.
(ii) \textsf{OpenDTR} outperforms \textsf{DPR} and \textsf{DPR-T}.
Interestingly, \textsf{DPR-T}, which incorporates special tokens to denote different table structures, performs worse than \textsf{DPR}, suggesting that simply adding special tokens does not enhance the encoder's understanding of tabular data.
(iii) The multi-vector dense method (\textsf{Solo-Ret}) does not consistently outperform single-vector dense methods. For instance, \textsf{Solo-Ret} surpasses \textsf{OpenDTR} on FetaQA but underperforms it on NQ-Tables.
This behavior can be explained by \textsf{Solo-Ret}'s fine-grained representation of cell-attributes-cell triplets, which tends to overemphasize local information of some specific triplets, but sometimes overlooks the important global semantics of the table.
(iv) \textsc{Birdie} outperforms all existing dense methods. For example, \textsc{Birdie} achieves an average improvement of 16.8\% in $P@1$ over the best-performing dense method on NQ-Tables and FetaQA. This demonstrates the effectiveness of \textsc{Birdie}'s differentiable search index, which addresses the limitations of the existing dense methods using traditional representation-index-search pipeline. 
(v) \textsc{Birdie} even outperforms the SOTA retrieval-rerank method \textsf{Solo}. This underscores \textsc{Birdie}'s superiority as an end-to-end solution, where beam search serves as an effective ranking mechanism.
(vi) All methods perform better on OpenWikiTable than the other two datasets. This can be attributed to: 1) the queries of OpenWikiTable contain detailed title-relevant information of ground truth tables, making it easier to locate the relevant table; and 2) a remarkable 99.8\% (24,634 out of 24,680) of tables in OpenWikiTable have unique titles, simplifying the task of identifying the correct table using title information alone.

\begin{table}[t]
\footnotesize
\centering
\caption{Performance of \textsc{Birdie} compared to other baselines.}
\vspace{-0.15in}
\label{tab:overall}
\renewcommand{\arraystretch}{1.2} 
\setlength{\tabcolsep}{0.7mm}{
\begin{tabular}{c|c|ccc|ccc|ccc} 
\specialrule{.08em}{.06em}{.06em}
\multicolumn{2}{c|}{\multirow{2}{*}{Methods}} & \multicolumn{3}{c|}{NQ-Tables} & \multicolumn{3}{c|}{FetaQA} & \multicolumn{3}{c}{OpenWikiTable} \\ 
\cline{3-11}
\multicolumn{2}{c|}{}                         & $P@1$ & $P@5$ & $T$(s)           & $P@1$ & $P@5$ & $T$(s)   & $P@1$ & $P@5$ & $T$(s)       \\ 
\hline
Sparse                 & \textsf{BM25}                 & 17.31 &31.07  &\textbf{0.053}    &45.38 &70.59 &0.034 &49.88 & 75.39 & \textbf{0.039}          \\ 
\hline
\multirow{4}{*}{Dense} & \textsf{DPR}                  & 23.15 & 51.30 &0.296                & 49.93 & 73.69 &0.024 & 64.31 & 87.59 & 0.046             \\
                       & \textsf{DPR-T}                & 13.24 & 36.50 &0.297                & 40.89 & 68.20 &0.022  &57.16 &84.22 &0.045          \\
                       & \textsf{OpenDTR}              & 37.37 & 64.28 &0.296                & 54.22 & 76.93 &\textbf{0.020} & 90.90 & 98.90 &   0.102         \\
                       & \textsf{Solo-Ret}             & 36.18 & 53.70 &3.532                & 79.28 & 88.72 &1.151 & 90.37 & 94.38 &  1.724           \\

\hline
Rerank          & \textsf{Solo}                 & 41.92 & 70.49 &4.398                & 81.43 & 91.51 & 1.857 & 94.71  & 96.59 &  2.541           \\
\specialrule{.08em}{.06em}{.06em}
Ours                  & \textsc{Birdie}           & \textbf{46.64} & \textbf{73.21} &0.097                  & \textbf{86.27} & \textbf{92.56} &0.145  & \textbf{97.20} & \textbf{99.06}   &0.105          \\ 
\specialrule{.08em}{.06em}{.06em}
\end{tabular}}
\vspace{-2mm}
\end{table}

\begin{table} \small
\centering
\caption{Memory/storage usage (GB) on NQ-Tables.}
\renewcommand{\arraystretch}{1} 
\vspace{-0.15in}
\label{tab:memory}
\setlength{\tabcolsep}{1.2mm}{
\begin{tabular}{c|ccccc} 
\specialrule{.08em}{.06em}{.06em}
Phases    & \textsf{DPR} & \textsf{OpenDTR}  & \textsf{Solo-Ret} & \textsf{Solo}  & \textsc{Birdie}  \\ 
\hline
\makecell{Training (GPU Memory)}       & 5.83        & 5.86  &   --   &   7.61  &   10.83   \\
 
\makecell{Indexing (Disk Storage)}        &  0.49    &  0.49   & 16.10   &  16.10   &  2.30    \\
 
\makecell{Inference (GPU Memory)}      &  0.97     &  2.18   & 1.46    & 7.27   & 2.33      \\
\specialrule{.08em}{.06em}{.06em}
\end{tabular}}
\vspace{-3mm}
\end{table}

\noindent \textbf{Efficiency.} In terms of online search efficiency, \textsf{BM25} demonstrates short time across all the datasets, as shown in Table~\ref{tab:overall}. However, for dense methods, search time increases with the size of the table repository.
The difference in search time is less pronounced for \textsf{Solo-Ret}, as its efficiency depends more on the number of triplets extracted from tables rather than the number of tables themselves.
\textsf{Solo} requires the most time due to its reranking process.
\textsc{Birdie} shows an average of $20\times$ and $27\times$ shorter search time on the three datasets compared with \textsf{Solo-Ret} and \textsf{Solo}, respectively. Note that,  \textsc{Birdie} maintains stable efficiency across varying sizes of table repositories, presenting a significant advantage.

\noindent \textbf{Memory/Storage Usage.} Table~\ref{tab:memory} presents the memory/storage usage during the key phases. For fairness, we set the batch size to 1 for all methods during training and inference. \textsc{Birdie} requires 10.83 GB of memory for training, more than baseline methods due to its larger encoder-decoder model. $\textsf{Solo-Ret}$, using a pre-trained model for retrieval, requires no extra training. In contrast, $\textsf{Solo}$ requires training the re-ranker.
For indexing, dense methods require storing
the table/triplet embeddings on disk, while Birdie requires only 2.3 GB to store the DSI model.
For inference, \textsc{Birdie} uses just 2.33 GB of memory, making it highly cost-efficient for deployment.

\noindent \textbf{Scalability.} 
We evaluate the scalability of clustering-based tabid assignment and DSI model training. 
For clustering, we implement mini-batch $k$-means with a batch size of 1k to construct the semantic tree for large repositories.
By expanding the NQ-Tables dataset through table duplication up to 1 million, we analyze memory usage with or without the mini-batch optimization, as shown in Figure~\ref{fig:tabid_scal}. The results indicate that memory usage scales linearly with the repository size, and mini-batch $k$-means effectively reduces memory requirements. 
For model training, memory requirements are directly influenced by batch size and can be reduced using gradient accumulation (GA)~\cite{GA}.
Figure~\ref{fig:memory_bsz} illustrates the GPU memory usage for varying batch sizes, including a scenario with a per-device effective batch size of 16 and GA applied.

\begin{figure}[t]
\vspace{-2mm}
\subfigure[CPU memory for tabid assignment]{
\label{fig:tabid_scal}
\includegraphics[width=0.44\linewidth]{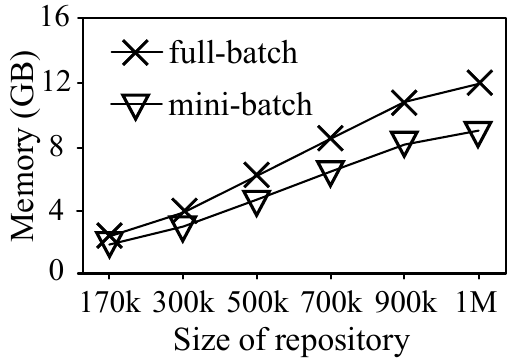}
} 
\subfigure[GPU memory for training]{
\label{fig:memory_bsz}
\includegraphics[width=0.44\linewidth]{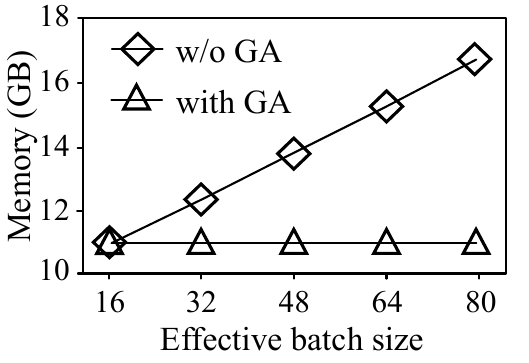}
}
\vspace{-5mm}
\caption{ Scalability of tabid assignment and model training.} 
\label{fig:inctime}
\vspace{-2mm}
\end{figure}

\subsection{Ablation Study}



 


We conduct an ablation study with results presented in  Table~\ref{tab:ablation}.

\noindent \textbf{Table ID Assignment.} First, we replace our semantic tabids with simple atom IDs (0, 1, . . . , $N-1$) for each table. 
This change  results in significant accuracy drop on NQ-Tables.
This is attributed to the atom IDs' lack of semantic context, which complicates the training (indexing) process. The smaller performance drops on FetaQA and OpenWikiTable can be attributed to (i)
their smaller repository sizes compared to NQ-Tables (170k), which reduce the indexing burden on the DSI model; and (ii) the test queries being more specific and semantically rich, making it easier to generate correct IDs and narrowing the performance gap between atom and semantic IDs. Next, we remove the instance view (view2) used in tabid generation, resulting in a decrease in accuracy across three datasets. This underscores the importance of instance data in capturing essential table information. Finally, we remove the metadata view (view1), leaving only the instance data for tabid generation. This change also leads to a decline in accuracy, confirming the critical role of metadata in understanding tabular data.

\noindent \textbf{Query Generation.} First, we replace our LLM-based query generator with the method proposed by \textsf{Solo}~\cite{Solo}, which generates SQL queries based on pre-defined SQL templates and subsequently converts them to NL queries using the pre-trained SQL2Text~\cite{seq2sql} model. This method results in significant accuracy drops on NQ-Tables and FetaQA. The decline is likely due to error accumulation across the two phases of query generation and the lack of fine-tuning of the SQL2Text model for the specific domain of our datasets.
Then, we apply the textual-based method docT5query~\cite{docT5query} to train a query generator using training splits, treating tables as flattened text. This approach results in an average decrease of 26\% in $P@1$ and 10\% in $P@5$ on NQ-Tables and FetaQA, further confirming the superiority of our LLM-based query generator specifically tailored to tabular data. Finally, removing the table sampling (TS) module from the LLM-based query generation process reduces accuracy across datasets, underscoring the importance of TS in generating high-quality synthetic queries.
Note that the performance drops on OpenWikiTable with these ablations are small due to the nature of its test queries. As long as the generated synthetic queries include title information, the accuracy remains high.

\begin{table} \small
\centering
\caption{P@1 and P@5 of \textsc{Birdie} and its variants.}
\label{tab:ablation}
\renewcommand{\arraystretch}{1}
\vspace{-0.15in}
\setlength{\tabcolsep}{0.75mm}{
\begin{tabular}{c|c|cc|cc|cc} 
\specialrule{.08em}{.06em}{.06em}
\multirow{2}{*}{Stages}                                                    & \multirow{2}{*}{Strategies} & \multicolumn{2}{c|}{NQ-Tables} & \multicolumn{2}{c|}{FetaQA} & \multicolumn{2}{c}{OpenWikiTable} \\ 
\cline{3-8}
                                                                           &                          & $P@1$ &  $P@5$                 & $P@1$ & $P@5$     & $P@1$ & $P@5$           \\ 
\hline
                --                                                           & \textsc{Birdie}                   & \textbf{46.64} & \textbf{73.21}                  & \textbf{86.27} & \textbf{92.56}    & \textbf{97.20} & \textbf{99.06}            \\ 
\hline
    \multirow{3}{*}{\begin{tabular}[c]{@{}c@{}}Tabid\\Assignment\end{tabular}} & Atom ID                &  17.75 &	42.54                   & 83.58 & 90.81      & 94.29 & 98.53         \\
                                                                           & w/o view2                &  44.12     &   70.48                     & 86.07 & 92.41  & 94.33 & 97.01              \\
                                                                           & w/o view1                &  45.16     &   70.59                     & 84.62 & 91.91  & 96.91 & 98.99             \\ 
\hline
\multirow{3}{*}{\begin{tabular}[c]{@{}c@{}}Query\\Generation\end{tabular}} & \textsf{Solo}-based                     &16.60     &31.83                      &    24.06   &   36.10  & 93.80  & 97.82     \\
                                                                           & Textual-based                  &  30.15     &    62.19                    & 72.78 & 87.52   & 95.12 & 98.30                   \\
                                                                           & w/o TS                  &  33.19     &    63.34                    & 58.76  & 74.34 &95.75 &98.23                     \\
\specialrule{.08em}{.06em}{.06em}
\end{tabular}}
\vspace{-2mm}
\end{table}



\begin{figure*}[t]
\centering
\includegraphics[width=0.40\linewidth]{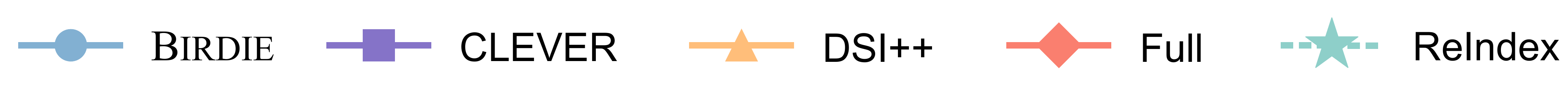}\\
\vspace{-2mm}
\subfigure[\scriptsize $AP@1$ on NQ-Tables]{
\begin{minipage}[t]{0.155\linewidth}
\includegraphics[width=1\linewidth]{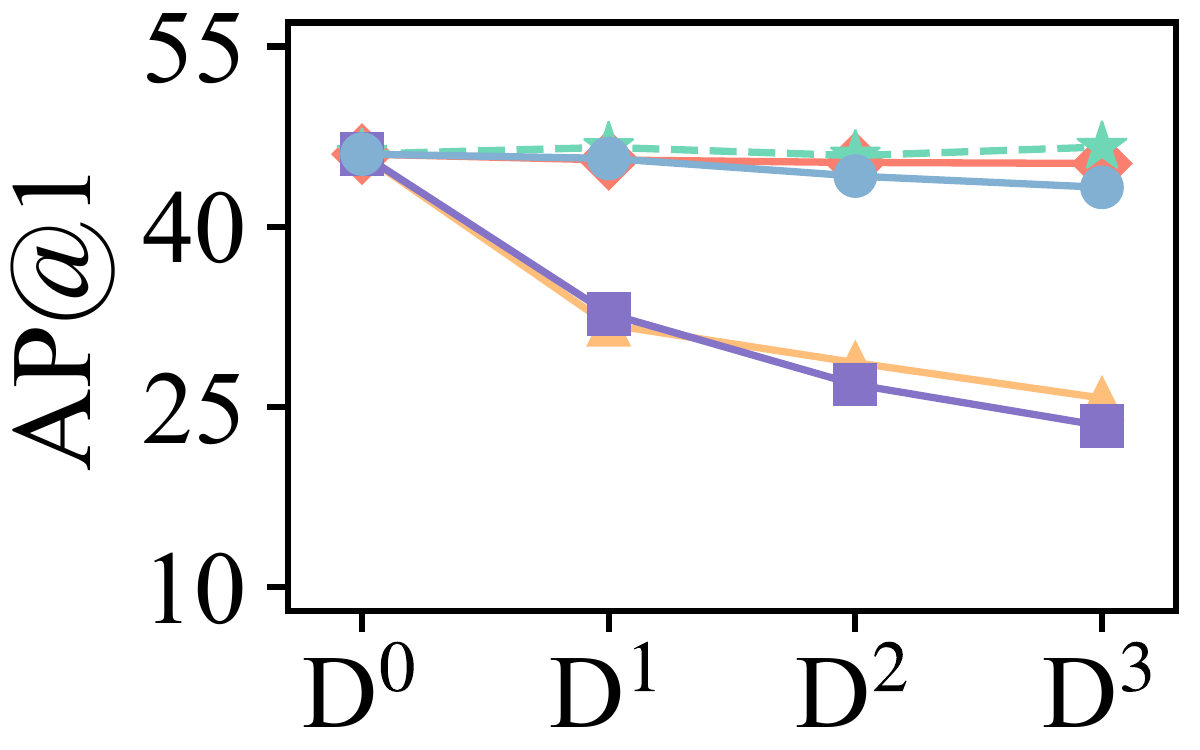 }
\end{minipage}%
}%
\subfigure[\scriptsize $FT@1$ on NQ-Tables]{
\begin{minipage}[t]{0.155\linewidth}
\includegraphics[width=1\linewidth]{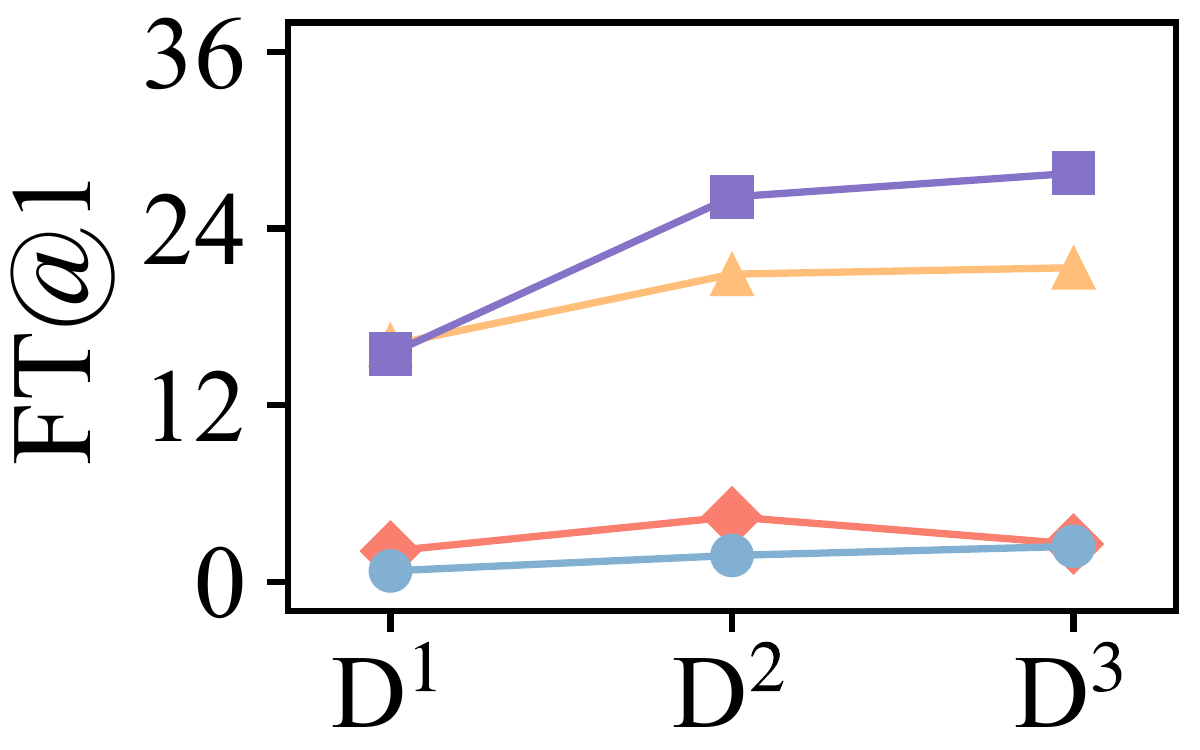}
\end{minipage}%
}%
\subfigure[\scriptsize $LP@1$ on NQ-Tables]{
\begin{minipage}[t]{0.155\linewidth}
\includegraphics[width=1\linewidth]{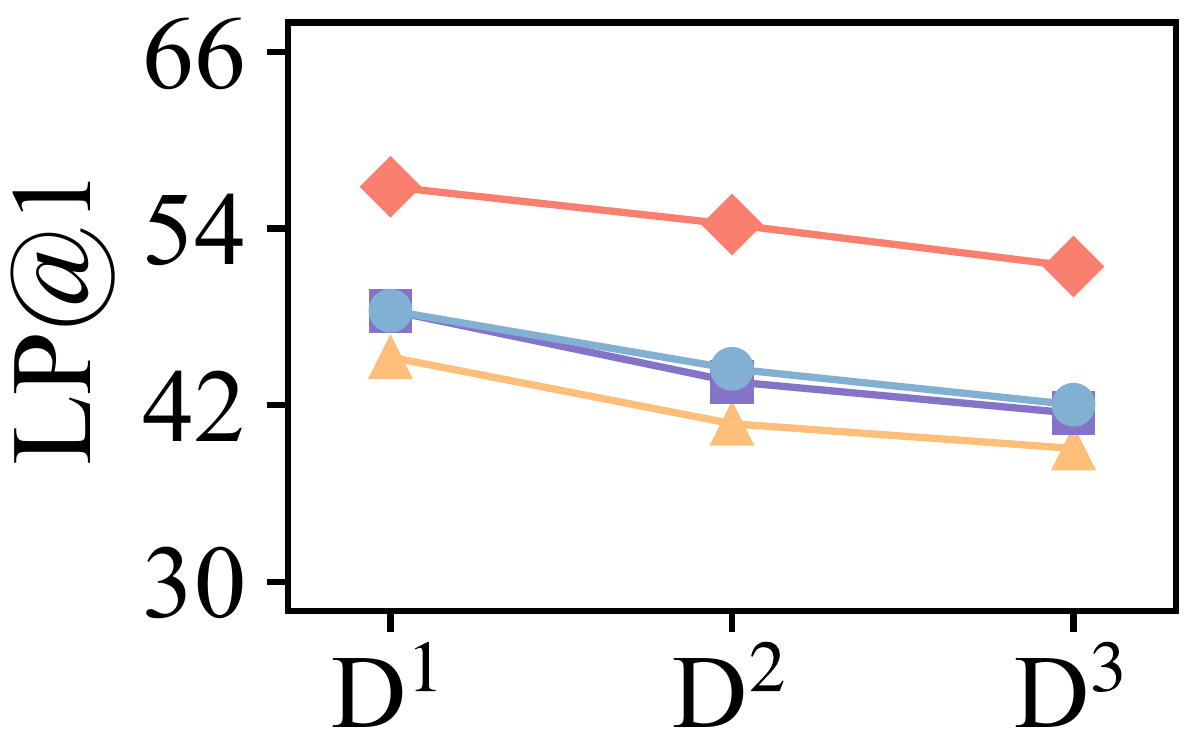}
\end{minipage}%
}%
\subfigure[\scriptsize $AP@5$ on NQ-Tables]{
\begin{minipage}[t]{0.155\linewidth}
\includegraphics[width=1\linewidth]{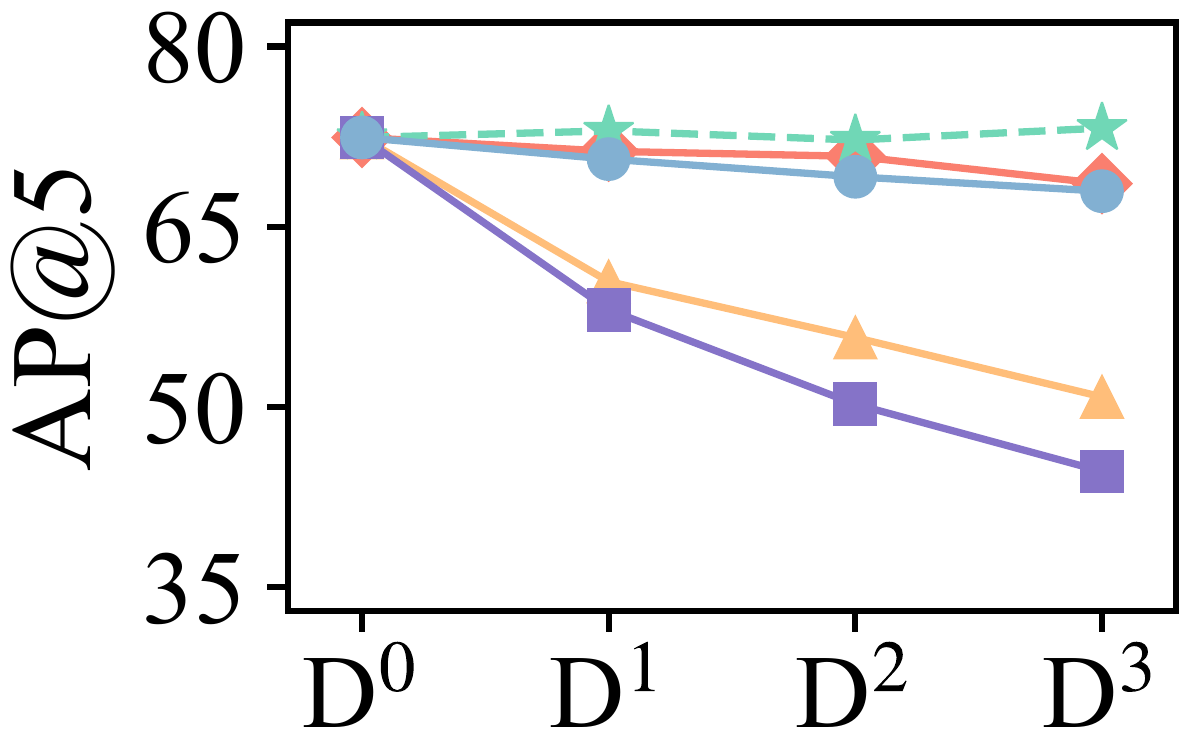}
\end{minipage}%
}%
\subfigure[\scriptsize $FT@5$ on NQ-Tables]{
\begin{minipage}[t]{0.155\linewidth}
\includegraphics[width=1\linewidth]{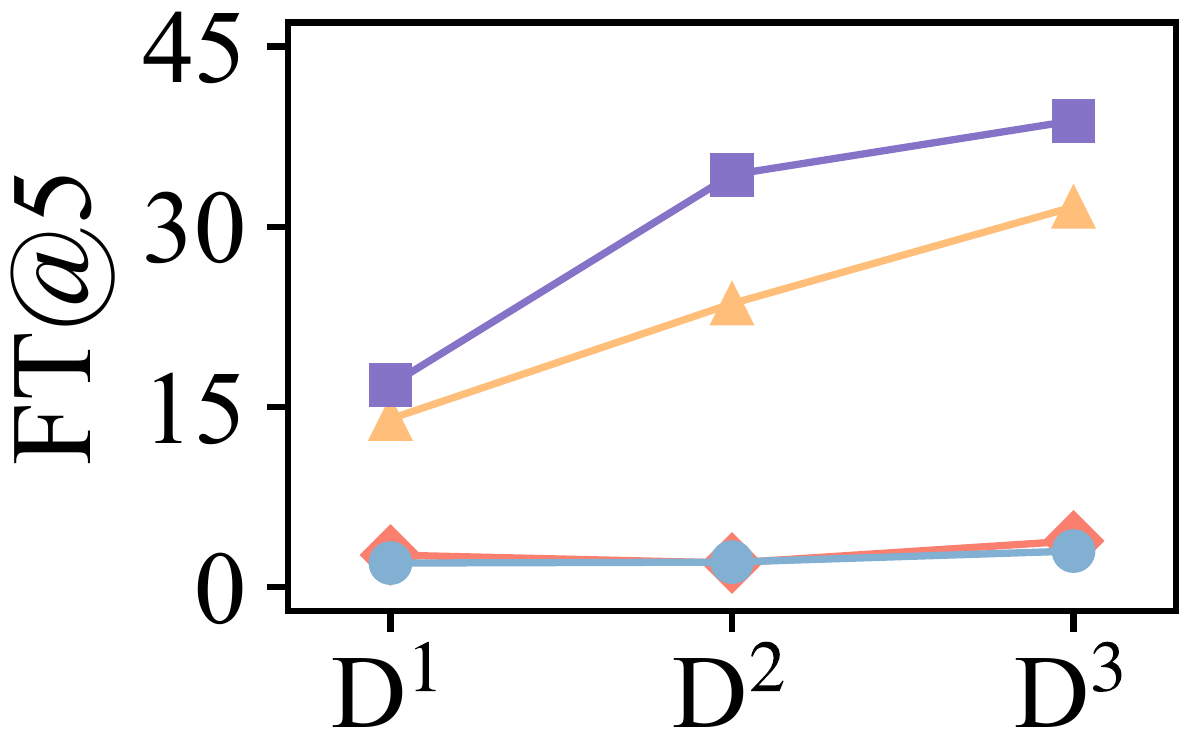}
\end{minipage}%
}%
\subfigure[\scriptsize $LP@5$ on NQ-Tables]{
\begin{minipage}[t]{0.155\linewidth}
\includegraphics[width=1\linewidth]{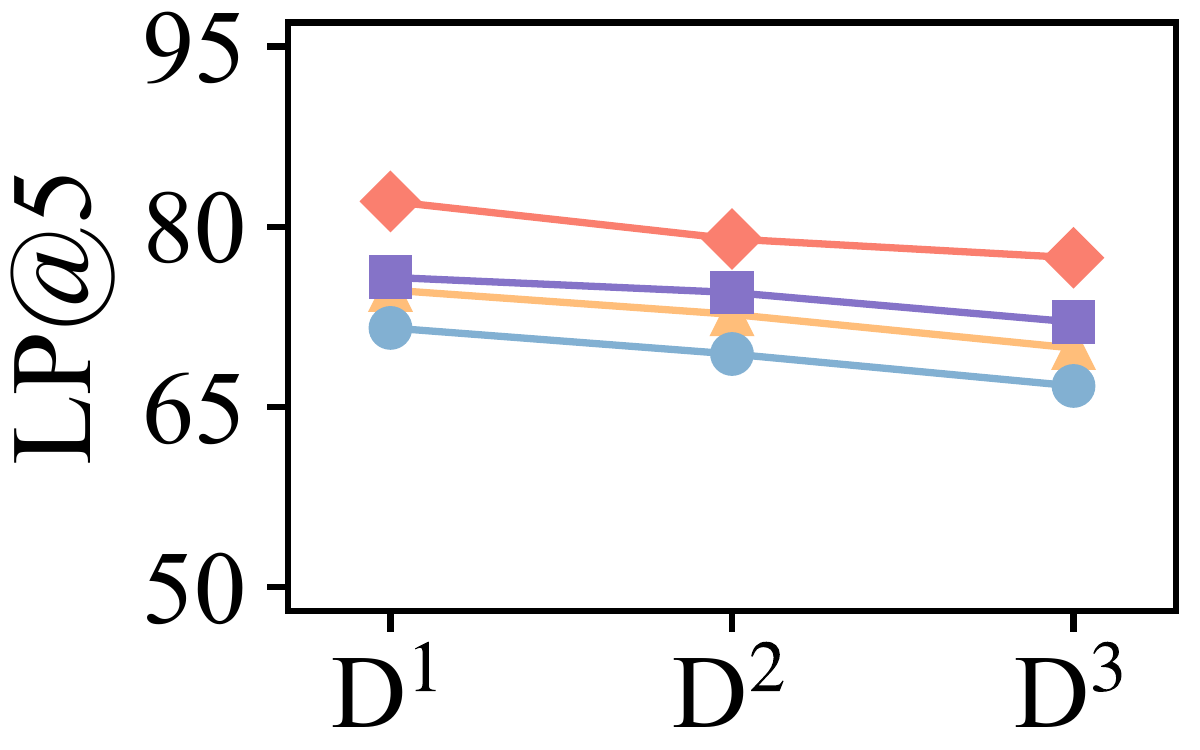}
\end{minipage}
}%
\\ \vspace{-2mm}
\subfigure[\scriptsize $AP@1$ on FetaQA]{
\begin{minipage}[t]{0.155\linewidth}
\includegraphics[width=1\linewidth]{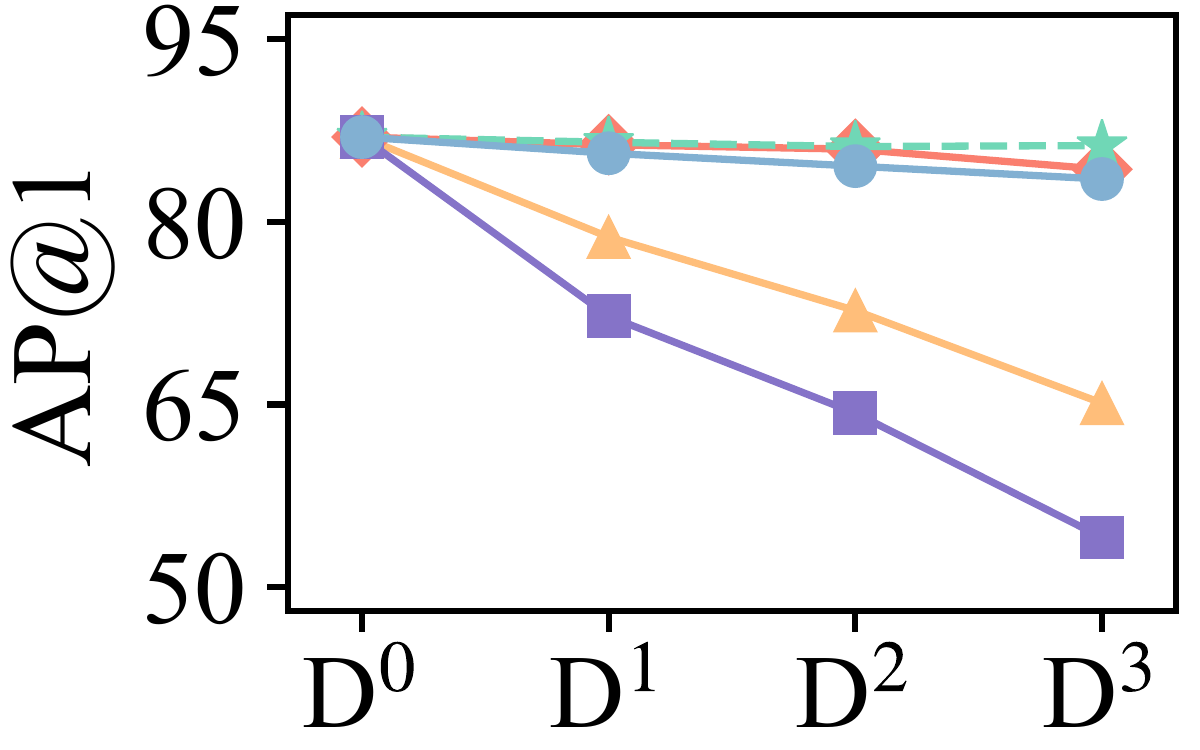 }
\end{minipage}%
}%
\subfigure[\scriptsize $FT@1$ on FetaQA]{
\begin{minipage}[t]{0.155\linewidth}
\includegraphics[width=1\linewidth]{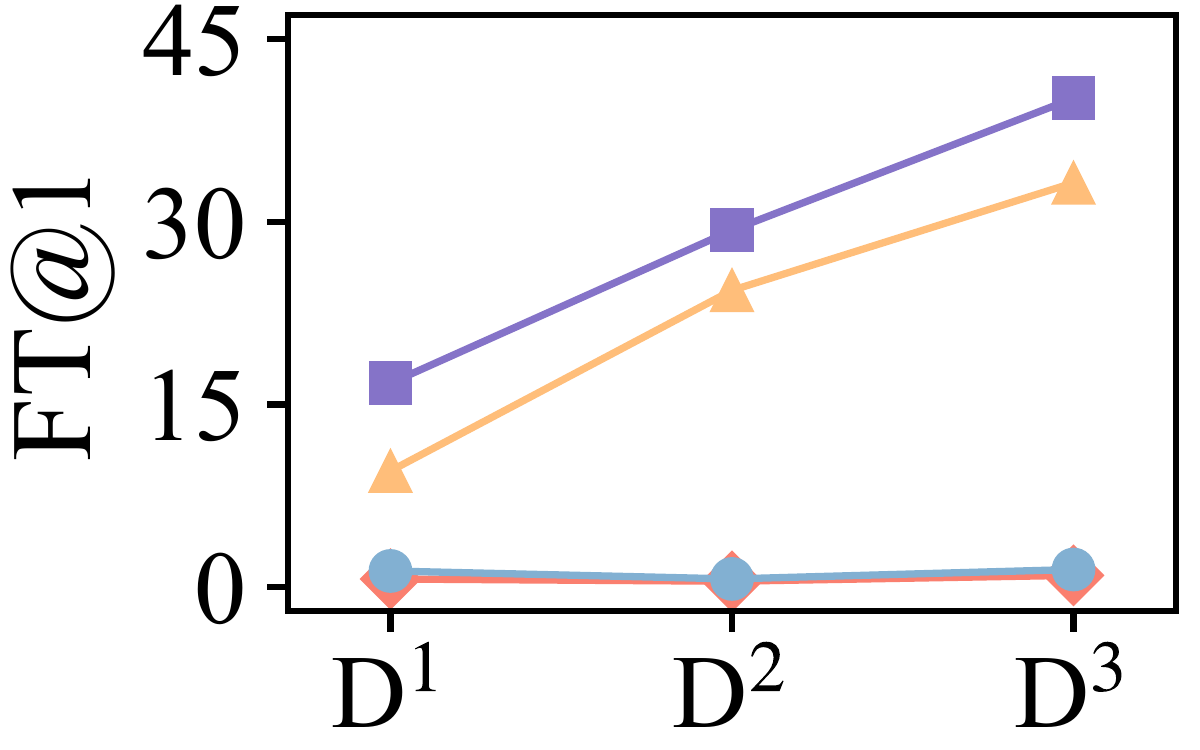}
\end{minipage}%
}%
\subfigure[\scriptsize $LP@1$ on FetaQA]{
\begin{minipage}[t]{0.155\linewidth}
\includegraphics[width=1\linewidth]{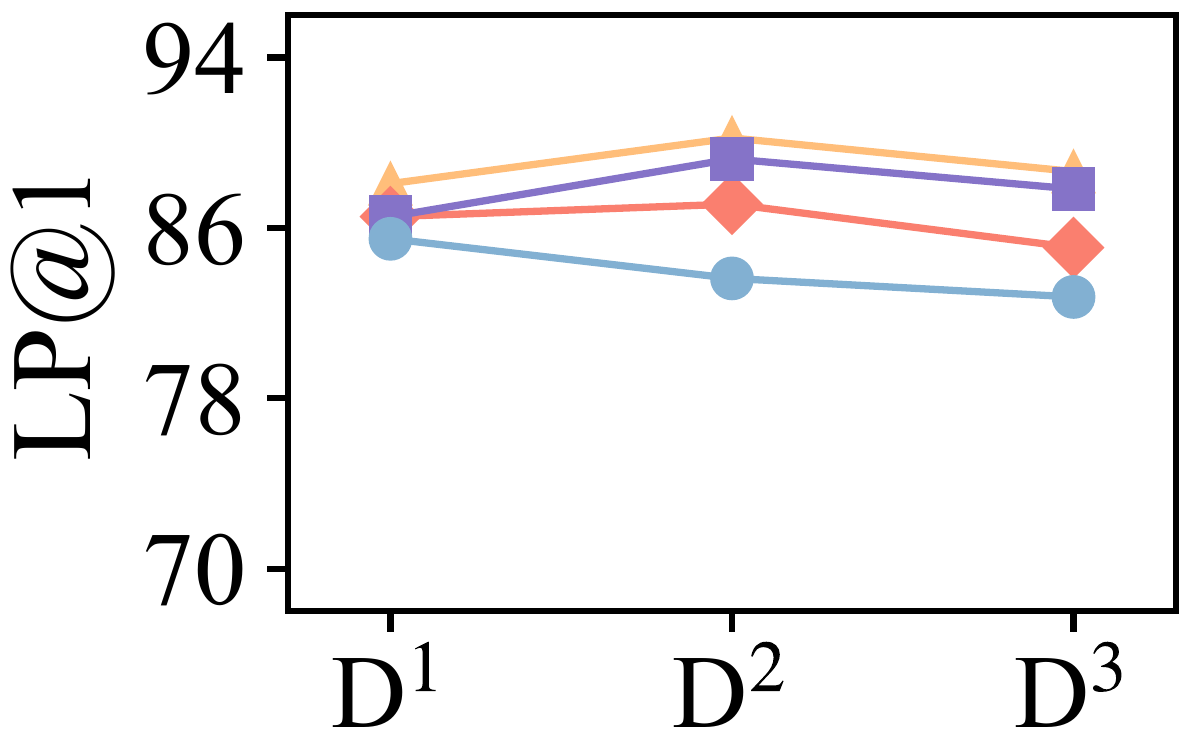}
\end{minipage}%
}%
\subfigure[\scriptsize $AP@5$ on FetaQA]{
\begin{minipage}[t]{0.155\linewidth}
\includegraphics[width=1\linewidth]{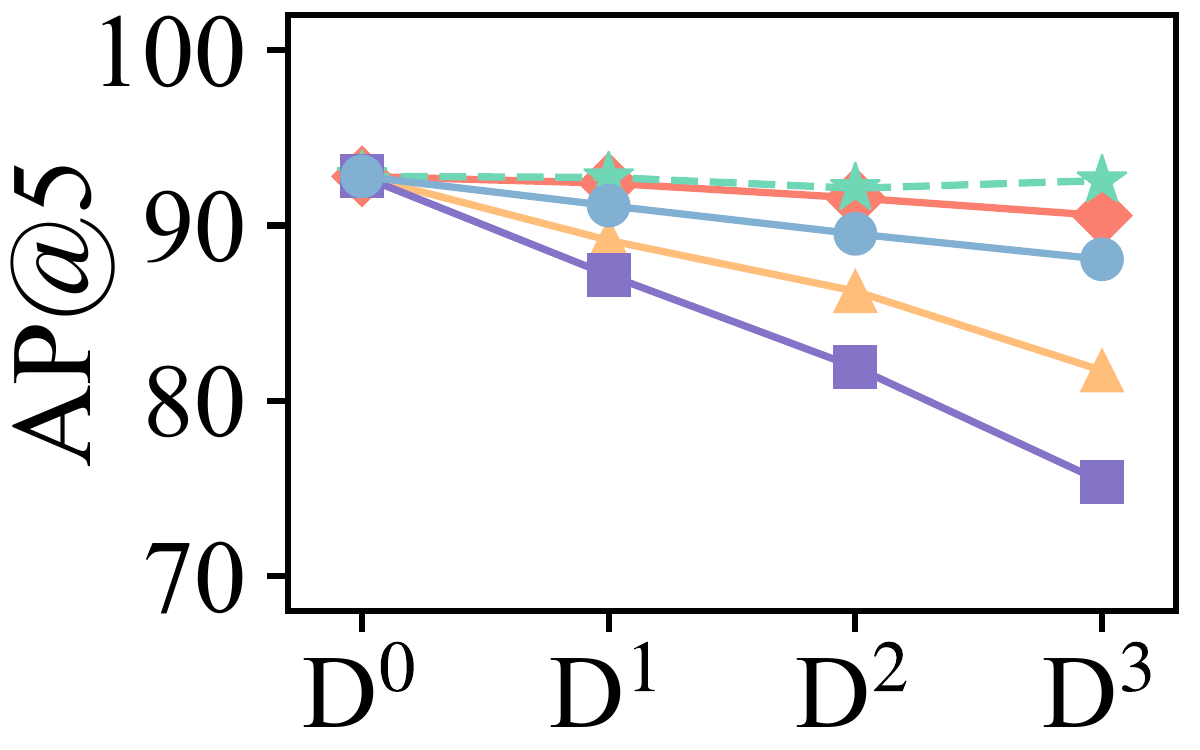}
\end{minipage}%
}%
\subfigure[\scriptsize $FT@5$ on FetaQA]{
\begin{minipage}[t]{0.155\linewidth}
\includegraphics[width=1\linewidth]{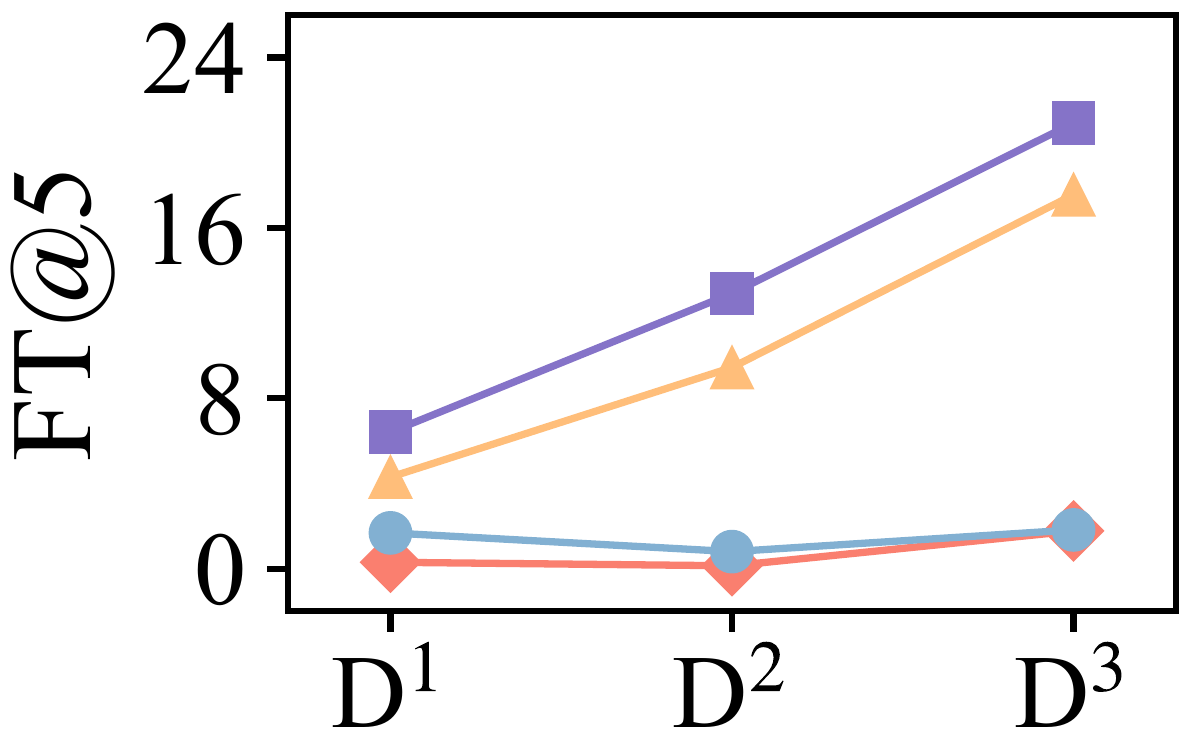}
\end{minipage}%
}%
\subfigure[\scriptsize $LP@5$ on FetaQA]{
\begin{minipage}[t]{0.155\linewidth}
\includegraphics[width=1\linewidth]{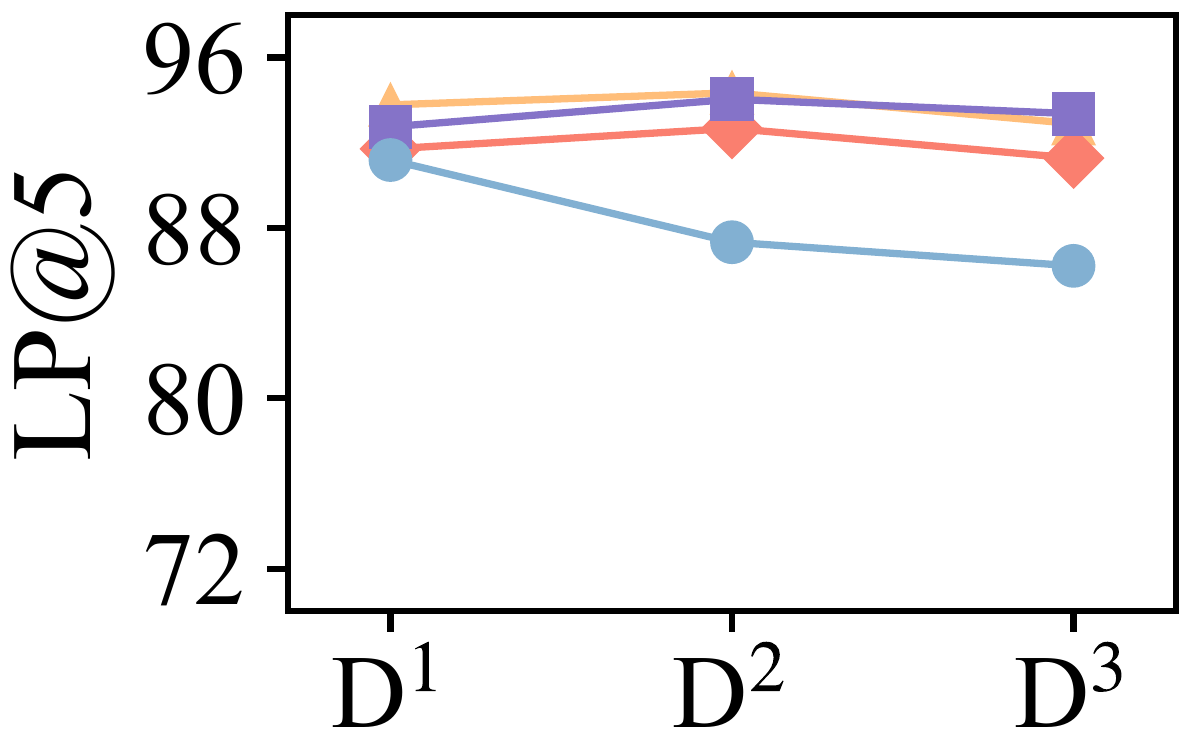}
\end{minipage}
}%
\\ \vspace{-2mm}
\subfigure[\scriptsize$AP@1$ on  OpenWikiTable]{
\begin{minipage}[t]{0.155\linewidth}
\includegraphics[width=1\linewidth]{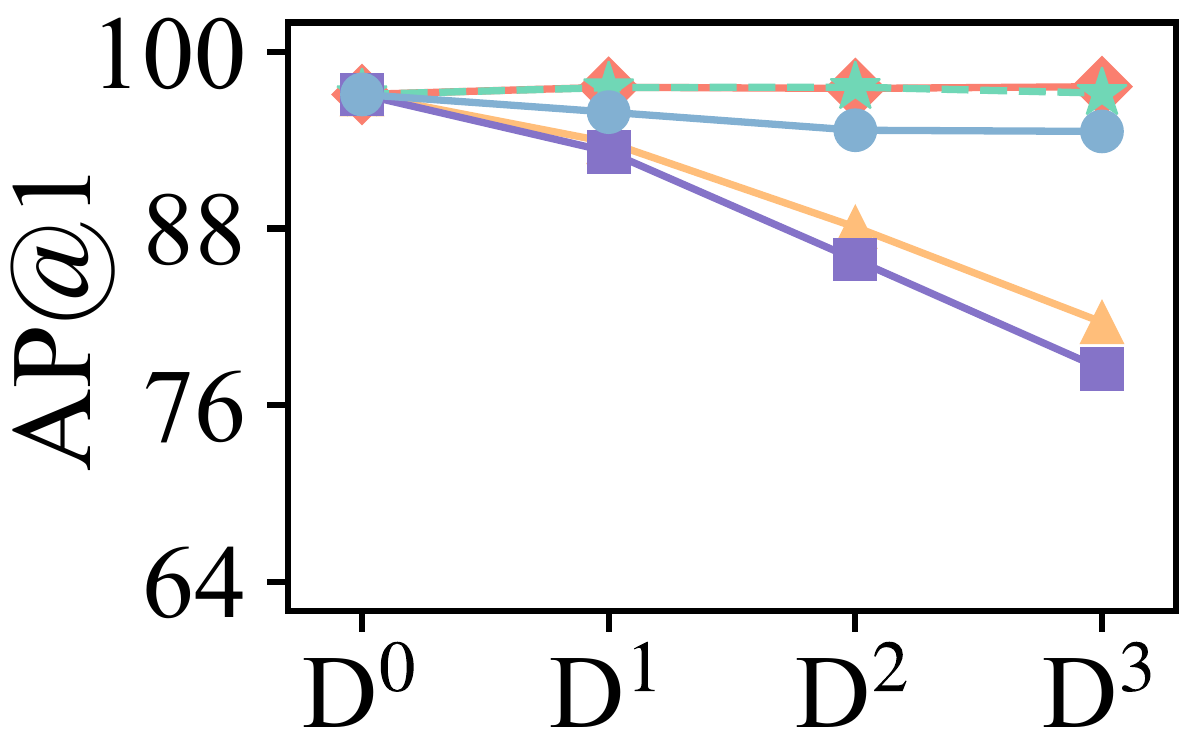 }
\end{minipage}%
}%
\subfigure[\scriptsize $FT@1$ on OpenWikiTable]{
\begin{minipage}[t]{0.155\linewidth}
\includegraphics[width=1\linewidth]{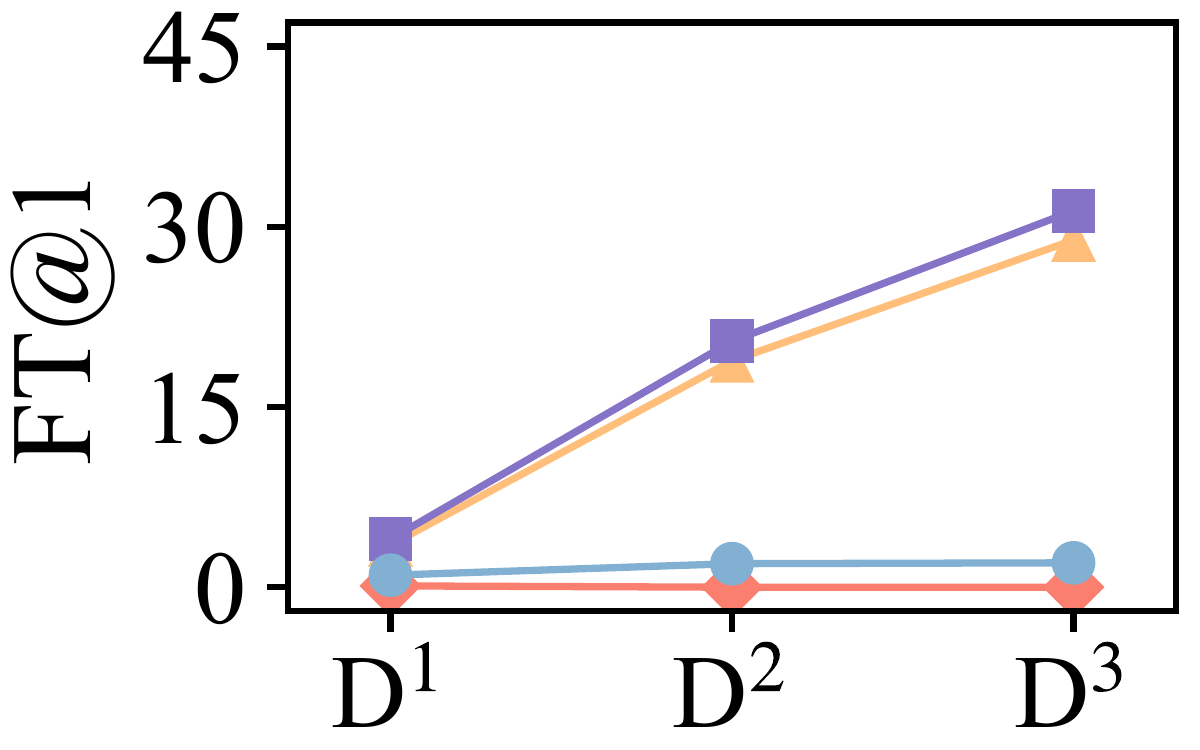}
\end{minipage}%
}%
\subfigure[\scriptsize $LP@1$ on OpenWikiTable]{
\begin{minipage}[t]{0.155\linewidth}
\includegraphics[width=1\linewidth]{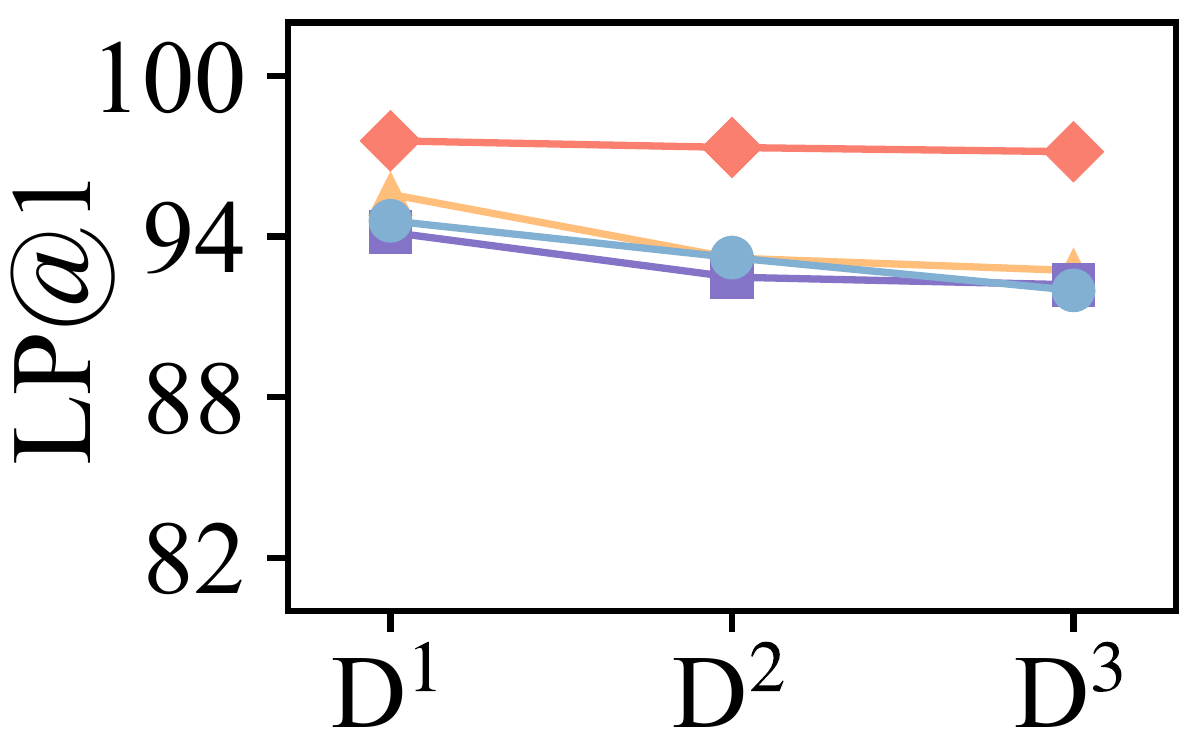}
\end{minipage}%
}%
\subfigure[\scriptsize $AP@5$ on OpenWikiTable]{
\begin{minipage}[t]{0.155\linewidth}
\includegraphics[width=1\linewidth]{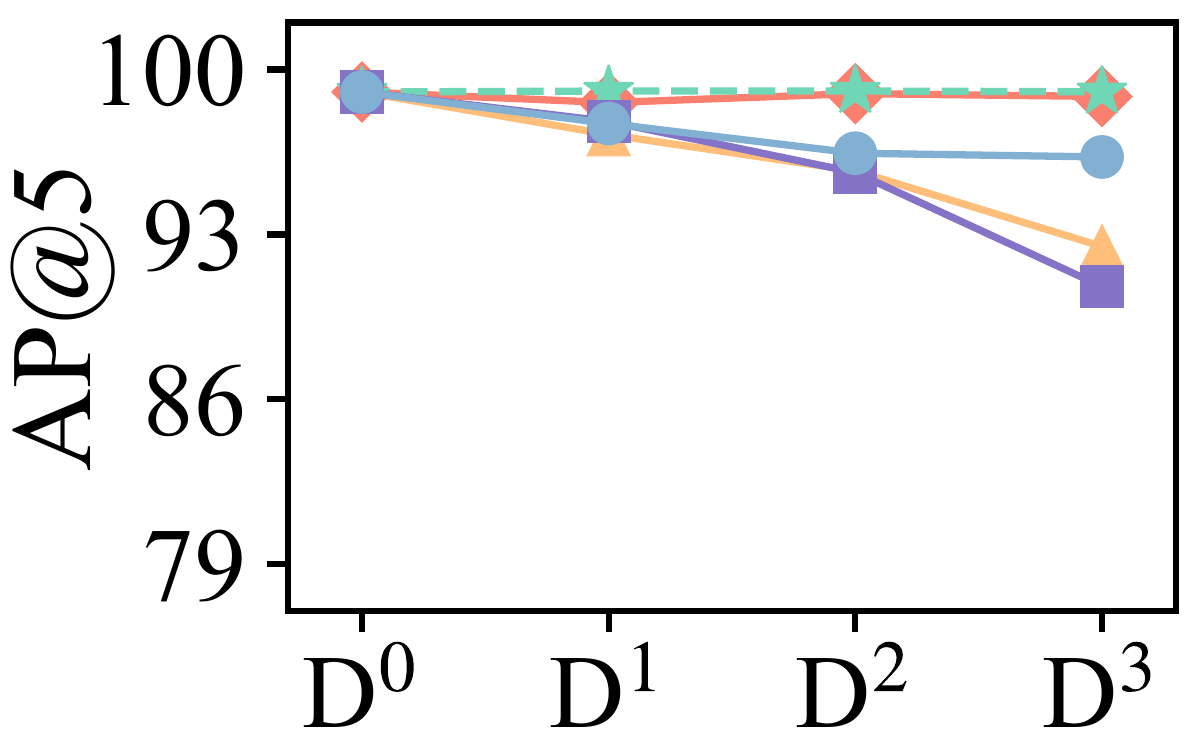}
\end{minipage}%
}%
\subfigure[\scriptsize $FT@5$ on OpenWikiTable]{
\begin{minipage}[t]{0.155\linewidth}
\includegraphics[width=1\linewidth]{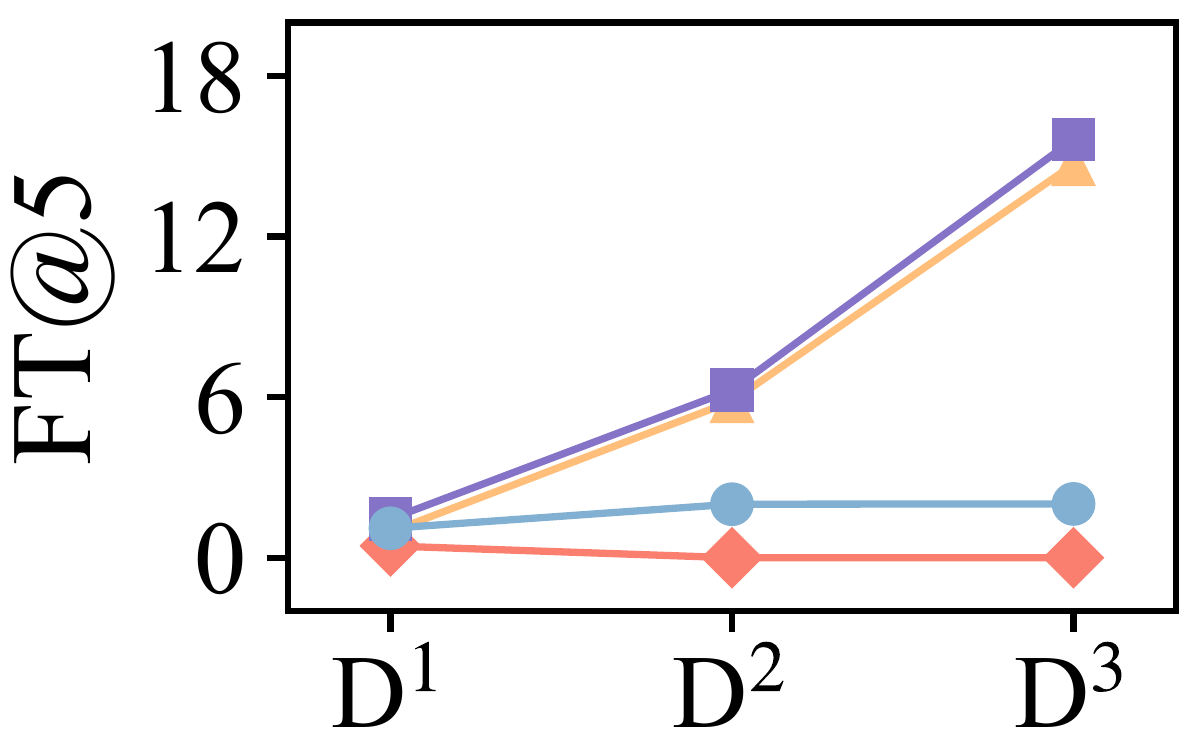}
\end{minipage}%
}%
\subfigure[\scriptsize $LP@5$ on OpenWikiTable]{
\begin{minipage}[t]{0.155\linewidth}
\includegraphics[width=1\linewidth]{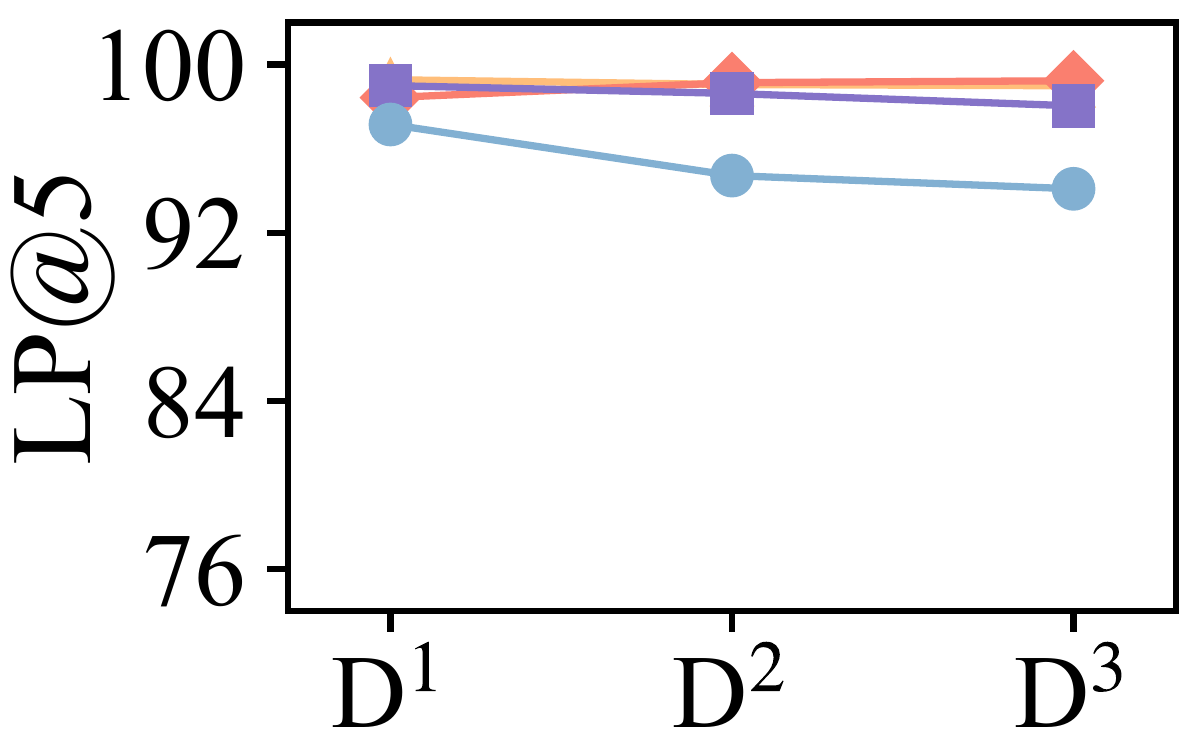}
\end{minipage}
}%
\vspace{-3mm}
\caption{Performance of index update under dynamic scenario.}
\vspace{-2mm}
\label{fig:p-c}
\end{figure*}

\subsection{Evaluation of Index Update}
\label{subsec:index_update}
\noindent \textbf{Setting.} To simulate a dynamic scenario, we randomly sample 70\% of the tables from the original dataset to create the initial repository $D^0$. We then randomly sample 10\% of the remaining tables to form a new batch, repeating this process three times to generate batches $D^1$, $D^2$, and $D^3$. For each dataset $D^i$ ($i\in[0,3]$), we construct the query set $Test^i$ by filtering queries from the original testing queries whose ground truth tables are included in $D^i$.

\noindent \textbf{Competitors.}
We adapt two existing continual learning methods for document retrieval via DSI: \textsf{CLEVER}~\cite{CLEVER} and \textsf{DSI++}~\cite{DSI++}, for table discovery. 
\textsf{CLEVER} samples old tables with similar tabids, and combines them with the new batch $D^u$ to continually train the DSI model. In contrast, \textsf{DSI++} randomly samples an equal number of old tables as \textsf{CLEVER}. Additionally, we compare a naive yet time-consuming solution, \textsf{Full}, which uses all tables in $\cup_{i=0}^{u} D^i$ during continual training. All three methods adopt our incremental tabid assignment strategy for new tables, and continually train the model from the last checkpoint. 
Lastly, we include \textsf{ReIndex}, which re-clusters all tables in $\cup_{i=0}^{u} D^i$ to reassign tabids, and re-trains the model  from scratch, serving as a theoretical upper bound.

\noindent \textbf{Metrics.} Following previous studies~\cite{CLEVER, DSI++}, we employ three metrics: (i) average performance ($AP@K$)  to measure the average performance on all existing tables, (ii) forgetting ($FT@K$) to measure the performance decline on old tables after learning from a new batch $D^u$; and (iii) learning performance ($LP@K$) to measure the ability to learn after indexing a new batch of tables $D^u$.
Formally, we denote  $P@K$ on $D^w$ (tested by $Test^w$) after indexing the new batch $D^u$ as $P_{u, w}@K$, and $P@K$ on $\cup_{w=0}^u D^w$ (tested by $\cup_{w=0}^uTest^w$) after indexing $D^u$ as $AP@K$.
The other two metrics are defined as follows:
$LP@K =\frac{1}{u} \sum\nolimits_{w=1}^u P_{w,w}@K$, and $FT@K=\frac{1}{u} \sum\nolimits_{w=0}^{u-1} \max _{w^{\prime} \in\{0, \cdots, u-1\}}\left(P_{w^{\prime}, w}@K-P_{u, w}@K\right)$.
%

\begin{figure}[t]
\vspace{-3mm}
\centering
 \subfigure{
\includegraphics[width=0.65\linewidth]{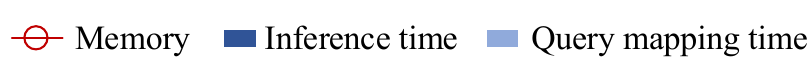}
}\\
\vspace{-4mm}
\subfigure[Serial]{
\includegraphics[width=0.44\linewidth]{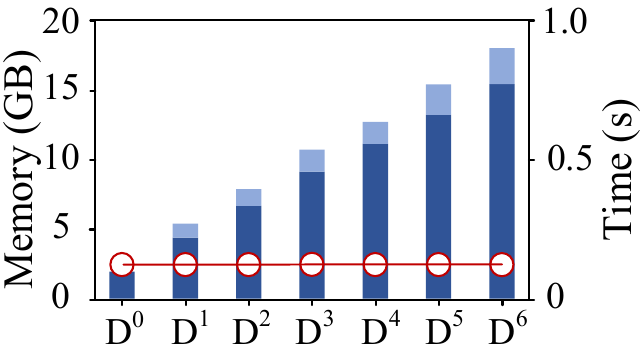}
}
\subfigure[Parallel]{
\includegraphics[width=0.44\linewidth]{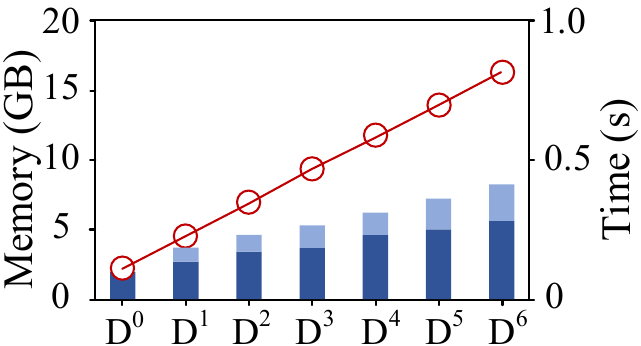}
}
\vspace{-5mm}
\caption{ Runtime and memory usage of online search.} 
\label{fig:scalability}
\vspace{-2mm}
\end{figure}

Note that, since \textsf{ReIndex} re-trains the model from scratch instead of continual learning, we only report its $AP$ values.

\noindent \textbf{Effectiveness.} The results are illustrated in Figure~\ref{fig:p-c}. In terms of $AP$ (the higher the better),  \textsc{Birdie} outperforms both \textsf{CLEVER} and \textsf{DSI++}, approaching the performance of \textsf{Full} and \textsf{ReIndex}, which use all previous tables. Interestingly, \textsf{DSI++}, which employs random sampling of old tables during continual learning, surpasses \textsf{CLEVER}, which utilizes a similarity-based sampling approach. Regarding $FT$ (the lower the better),
\textsc{Birdie} reduces forgetting by over 90\% compared to \textsf{CLEVER} and \textsf{DSI++}.
Additionally, as new batches are introduced, forgetting intensifies for \textsf{CLEVER} and \textsf{DSI++}, since each model update can lead to the forgetting of some old tables, accumulating with ongoing parameter updates. In contrast, \textsc{Birdie} implements parameter isolation, which keeps memory units independent, resulting in reduced forgetting. For $LP$ (the higher the better),
\textsc{Birdie} demonstrates comparable $LP@1$ and slightly lower $LP@5$ compared to  other  continual-learning-based methods.
The higher $LP$ of \textsf{CLEVER} and \textsf{DSI++} indicates that replaying only a portion of old tables during continual learning tends to favor new table retention. 
This prioritization of new tables leads to increased forgetting, ultimately resulting in poorer average performance. In contrast, \textsc{Birdie} strikes a promising balance between $FT$ and $LP$, yielding higher $AP$.  Note that $LP$ measures accuracy in an extreme scenario where the workload consists solely of new memory retrieval, implying that the older data may have become outdated. In Appendix~\ref{appendix:C.2}, we present a simple extension to improve \textsc{Birdie}'s performance as old memories become obsolete.

\begin{figure}[t]
\centering
\quad \quad \includegraphics[width=0.6\linewidth]{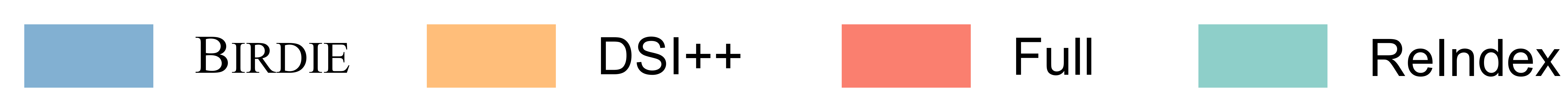}\\
\vspace{-2mm}

\subfigure[Tabid assignment time]{\begin{minipage}[t]{0.42\linewidth}
\includegraphics[width=1\linewidth]{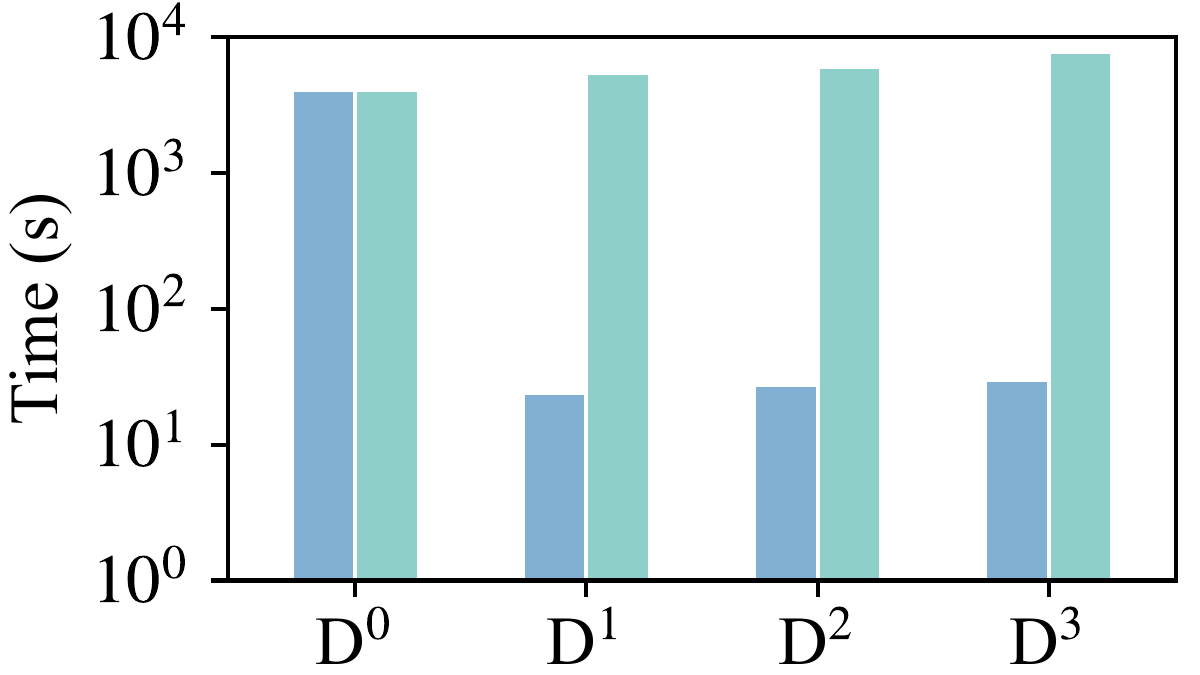}
\end{minipage}
}
\subfigure[Training time]{\begin{minipage}[t]{0.42\linewidth}
    \includegraphics[width=1\linewidth]{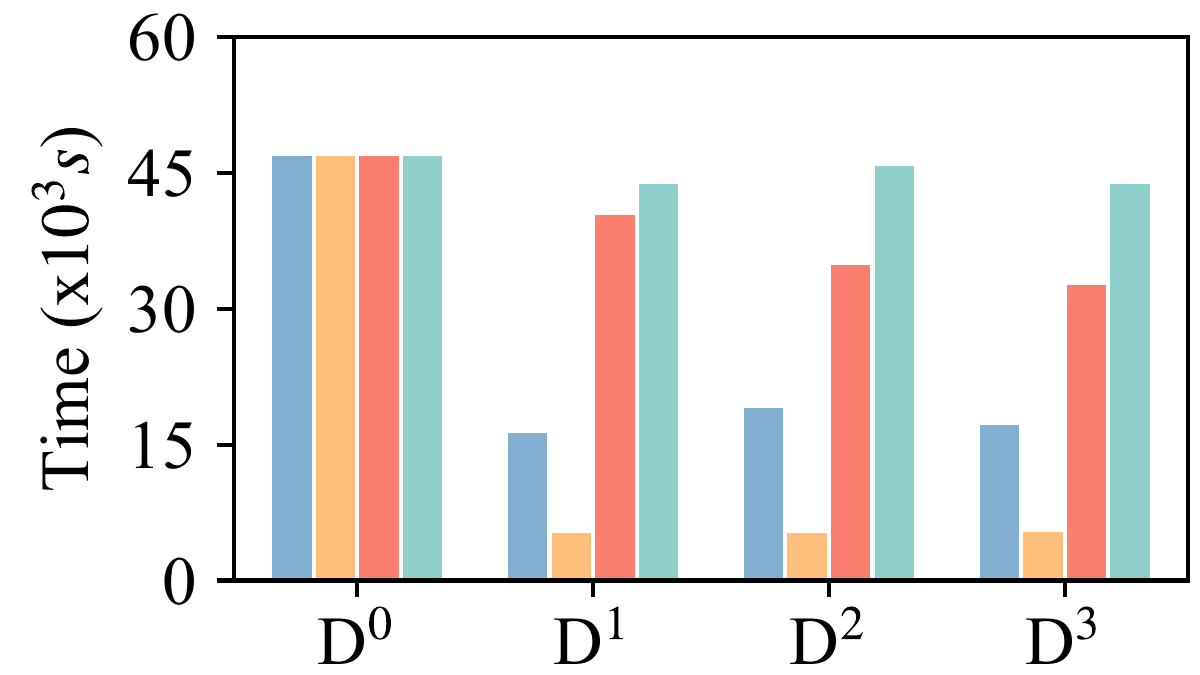}
\end{minipage}
}
\vspace{-5mm}
\caption{Runtime of index update on NQ-Tables.}
\label{fig:inc_time}
\vspace{-2mm}
\end{figure}

\noindent \textbf{Runtime.}  Figure~\ref{fig:inc_time}  illustrates the time taken for index update across sequential batches $D^0$ to $D^3$ on NQ-Tables. Runtime for other datasets can be found in Appendix~\ref{appendix:C.1}. For tabid assignment, \textsf{CLEVER},  \textsf{DSI++}, and \textsf{Full} follow \textsc{Birdie}, which explains why we only report the time taken by \textsc{Birdie} and \textsf{ReIndex}. It is observed that \textsc{Birdie}'s running time is 1-2 orders of magnitude shorter than that of \textsf{ReIndex}. For training time, \textsf{DSI++} is the fastest (the runtime of \textsf{CLEVER} is similar and omitted), utilizing less data than \textsf{Full} and \textsf{ReIndex} and fine-tuning the model from the last checkpoint, though its effectiveness is limited. \textsc{Birdie} demonstrates greater efficiency than \textsf{Full} and \textsf{ReIndex}, while achieving comparable average performance.

\begin{figure*}
\vspace{-10mm}
  \centering
  \includegraphics[width=1\linewidth]{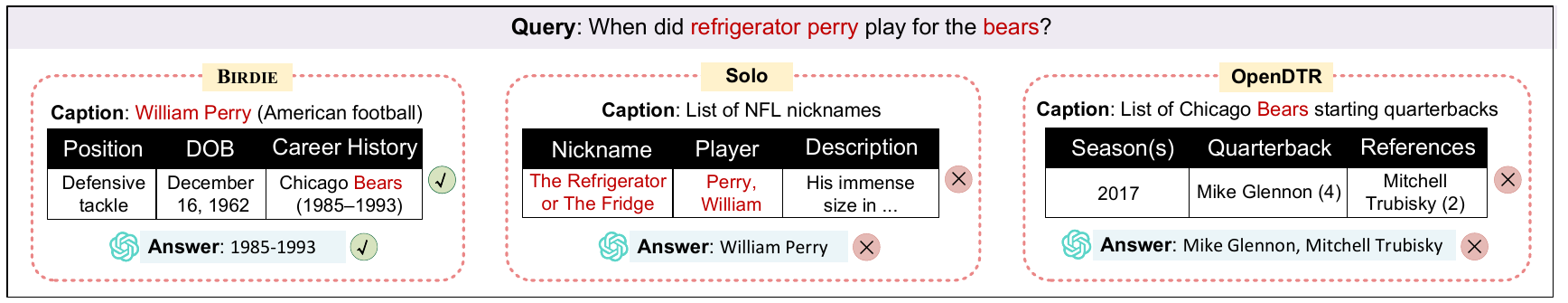}\vspace{-3mm}
  \caption{{A case of NL-driven table discovery and the results of table QA using GPT-4o.}
  \label{fig:case}}
    \vspace{-2mm}
\end{figure*}

\begin{table}
\small
\centering
\caption{Memory usage (GB) on NQ-Tables for index update.}
\renewcommand{\arraystretch}{1} 
\vspace{-0.15in}
\label{tab:memory_inc}
\setlength{\tabcolsep}{2.6mm}{
\begin{tabular}{c|cccc } 
\specialrule{.08em}{.06em}{.06em}
 Phases     & $D^0$ & $D^1$ & $D^2$ &  $D^3$    \\ 
\hline
ITA (CPU Memory)   &   1.31      & 0.46   &    0.48   &  0.51   \\

ITA w/o TC (CPU Memory)    &  1.31       & 2.16    &   2.52  &  2.74    \\

Training (GPU Memory)  &   15.19     &  8.04   &    8.04  &  8.04    \\

\hline
Inference-Serial (GPU Memory)  &   2.33      &  2.33    &  2.33     &  2.33    \\ 
Inference-Parallel (GPU Memory)  &   2.33      & 4.67    & 7.00     & 9.47    \\ 
\specialrule{.08em}{.06em}{.06em}
\end{tabular}}
\vspace{-5mm}
\end{table}

\noindent \textbf{Memory Usage.} Table~\ref{tab:memory_inc} reports memory usage on NQ-Tables.
ITA w/o TC refers to ITA without TC optimization presented in Section~\ref{subsec:clustering}. Results show that our TC optimization effectively reduces memory usage by 5$\times$. For training, we evaluate memory usage with the default batch size of 64. The training process on $D^1$–$D^3$ is memory-efficient with LoRA techniques. Depending on available memory resources, \textsc{Birdie} supports two inference paradigms: the serial paradigm, which requires constant memory, and the parallel paradigm, which scales linearly with the number of updates. 

\noindent \textbf{Scalability.} To evaluate the scalability during online search,
we use the NQ-Tables dataset to create six incremental batches.
Specifically, we sample 40\% of the tables from NQ-Tables as $D^0$, and the remaining 60\% tables are evenly divided into six batches $\{D_i\}^6_{i=1}$.
Figure~\ref{fig:scalability} illustrates the memory usage and runtime. 
In low-resource scenarios with serial inference, memory usage remains constant, while runtime scales linearly with the number of models.
Conversely, in rich-resource scenarios with parallel inference, runtime remains relatively stable, while memory usage scales linearly with the number of models. This is because multiple copies of the base model are loaded into memory simultaneously, as discussed in Section~\ref{subsec:parameter_isolaiton}.
The number of synthetic queries per candidate table is a small constant $B$ (e.g., 20 in our implementation), and each model generates only a limited number of candidate tables, making query mapping process both time and memory efficient.


\section{Case Study}
\label{sec:case}
We present a case of using \textsc{Birdie}, \textsf{Solo}, and \textsf{OpenDTR} to retrieve the top table from NQ-Tables in response to a given NL query. GPT-4o is then employed as a reasoning tool through the OpenAI API to answer the query based on the contents in the retrieved table. Figure~\ref{fig:case} illustrates
the case. Note that we only show the row relevant to the answer due to space limitation.

It is observed that only \textsc{Birdie} successfully retrieves the correct table, and GPT-4o provides the correct answer only when it receives the correct table.
\textsf{Solo} encodes each cell-attribute-cell triplet as an embedding and measures similarity between the query and triplet embeddings. In this instance, the triplet in the table ``List of NFL nicknames'' exhibits high similarity to the query. However, it lacks time information, rendering it incapable of answering the user's query. This highlights the limitation of dense methods that do not facilitate deep query-table interactions.
\textsf{OpenDTR}, which encodes the entire table as a single vector, struggles to capture detailed table information. As a result, the table ``Chicago Bears'' shows a higher similarity to the query than the correct table.
In contrast, \textsc{Birdie} optimizes both the indexing and search processes jointly and engages in deep query-table interactions during model training. Therefore, it comprehensively understands the given NL query and successfully identifies the correct table.
\section{Related Work}
\label{sec:relatedwork}

\noindent
\textbf{{Table Discovery. }}
Table discovery has been extensively researched within the data management community~\cite{TabelDiscovery,DataLake_Survey}. A prevalent line of table discovery is query-driven discovery,  which includes: (i) keyword-based table search~\cite{AdelfioS13,GoogleSearch} that aims to identify web tables related to specified keywords, utilizing metadata such as table headers and column names; 
(ii) table-driven search which locates target tables within a large data lake that can be joined~\cite{JOSIE,Deepjoin,Snoopy} or unioned~\cite{starmine,santos,TUS} with a given query table; and (iii) NL-query-driven table search~\cite{Solo,OpenDTR,OpenWiki}.

NL-driven table discovery offers a user-friendly interface that allows users to express their needs more precisely.
Existing NL-driven table discovery methods~\cite{Solo,OpenDTR,OpenWiki} typically follow a traditional representation-index-search pipeline.
The encoder in the representation phase plays a crucial role in search accuracy. For instance, OpenDTR~\cite{OpenDTR} uses TAPAS~\cite{TAPAS} as the backbone for its bi-encoder.
OpenWikiTable~\cite{OpenWiki} offers various options for query and table encoders. However, representing a table as a single vector can sometimes be insufficiently expressive. To address this, Solo~\cite{Solo} encodes each cell-attributes-cell triplet within the table into a fixed-dimensional embedding and retrieves similar triplet embeddings to the query embedding, followed
by aggregation of triplets-to-table. However, the lack of deep query-table interactions during retrieval hinders further performance improvements.

Another line of NL-based table search literature focuses on the re-ranking~\cite{GTR,AdHoc_TR,TableSearch}.
Utilizing a cross-encoder, they input both the query and candidate table to obtain embeddings for each query-table pair. This process enhances accuracy due to the deep query-table interactions but lacks of scalability for first-stage retrieval.

\vspace{1mm}
\noindent
\textbf{Differentiable Search Index.}
Differentiable search index (DSI)~\cite{DSI} sparks a novel search paradigm that unifies the indexing and search within a single Transformer architecture. It was initially proposed for document retrieval~\cite{NCI, DSI, DSI-QG} and has been applied in scenarios like retrieval-augmented generation (RAG)~\cite{CorpusLM}, recommendation systems~\cite{Tiger}, etc.  To the best of our knowledge, \textsc{Birdie} is the first attempt to perform table discovery using DSI, taking into account the unique properties of tabular data to automate the collection of training data. Real-world applications often involve dynamically changing corpora. However, in DSI, which encodes all corpus information into model parameters, indexing new corpora inevitably leads to the forgetting of old ones. To mitigate catastrophic forgetting, some recent studies~\cite{DSI++, CLEVER} propose replay-based solutions that sample some old data and combine it with new data for continual learning. However, these methods often struggle to balance indexing new data and retaining old memories, resulting in suboptimal average performance.
In contrast, \textsc{Birdie} designs a parameter isolation method that ensures the independence of each memory unit, thus achieving a promising average performance.

\section{Conclusions}
\label{sec:conlusion}

We present \textsc{Birdie}, an effective NL-driven table discovery framework using a differentiable search index. We first introduce a two-view-based tabid assignment method to assign a unique table identifier to each table, considering the semantics of both the metadata and instance data of tables. Then, we propose a LLM-based query generation method tailored for tabular data to construct synthetic NL queries for DSI model training. To accommodate the continual indexing of dynamic tables, we design an index update strategy via parameter isolation. Comprehensive experiments confirm that \textsc{Birdie} significantly outperforms the SOTA  methods, and our parameter isolation strategy alleviates catastrophic forgetting and achieves better average performance than competitors. In the future, we plan to develop acceleration methods to speed up the offline training phase of \textsc{Birdie} and reduce costs.




\bibliographystyle{ACM-Reference-Format}

\balance
\bibliography{refer}
\clearpage
\newpage
\appendix
\section*{\huge{Appendix}}
\vspace{0.1in}
\setcounter{myLemma}{0}
\renewcommand\themyProp{A.\arabic{myLemma}}

\section{Proof of Lemma}
\label{appendix:A}
\begin{myLemma}
The upper and lower bounds of the cluster radius $r'$ after the insertion of $\mathbf{h}_{new}$ are given by $r + \operatorname{dist}(\mathbf{c}, \mathbf{c}^\prime)$ and $\operatorname{dist}(\mathbf{h}_{new}, \mathbf{c}^\prime)$, respectively. 
\end{myLemma}
 
\begin{proof}
    After inserting $\mathbf{h}_{new}$ into the cluster, the farthest embedding from the new center $\mathbf{c}^\prime$ could either be an old embedding $\mathbf{h}_{old}$ within this cluster, or the newly inserted  $\mathbf{h}_{new}$. In the first case, applying the triangle inequality yields $r^\prime = \operatorname{dist}(\mathbf{h}_{old}, \mathbf{c}^\prime) \leq \operatorname{dist}(\mathbf{h}_{old}, \mathbf{c}) + \operatorname{dist}(\mathbf{c}, \mathbf{c}^\prime) 
    \leq r + \operatorname{dist}(\mathbf{c}, \mathbf{c}^\prime)$. In the second case, we can also infer that $r^\prime = \operatorname{dist}(\mathbf{h}_{new}, \mathbf{c}^\prime) \leq \operatorname{dist}(\mathbf{h}_{new}, \mathbf{c}) + \operatorname{dist}(\mathbf{c}, \mathbf{c}^\prime) = d + \operatorname{dist}(\mathbf{c}, \mathbf{c}^\prime)$. Since $ \ \delta < d \leq  r$,  we have $r' \leq  r + \operatorname{dist}(\mathbf{c}, \mathbf{c}^\prime)$. 
    
    The distance from the newly inserted $\mathbf{h}_{new}$ to the new center $\mathbf{c}^\prime$ provides a minimum bound for $r'$.
\end{proof}

\section{Sensitivity Study}
\label{appendix:B}

We provide some practical guidelines for the key hyperparameters and perform a sensitivity study.

For model training, most hyperparameters are set based on either default recommended values or commonly used configurations from previous studies. For example, the rank $d_r$ of LoRA is set to 8 following the default setting in PEFT library~\cite{peft}; the learning rate and beam size are chosen based on previous studies~\cite{DSI, DSI-QG}. For query generation, the size $B$ of queries for each table is a tunable hyperparameter. We explore the impact of the number $B$ of generated queries per table on \textsc{Birdie}'s performance by varying $B$ in the range of  \{5, 10, 20, 50\}.  
Figure~\ref{fig:vary_q} illustrates the performance changes over the training steps. As expected, generating more queries generally improves accuracy by providing more comprehensive coverage of table information. However, performance does not continuously increase with the number of queries; for instance, generating 20 queries yields similar results to generating 50. Thus, we generate 20 queries for each table by default.

During tabid assignment, key hyperparameters include the number $k$ of clusters, the maximum size $c$ of leaf clusters, and the depth $l$ of the first view. Recall that each tabid is represented as $\mathbf{s} =( s_1, \dots s_d)$, where $s_i \in [0, k-1]$ for $i<d$ and $s_d \in [0, c-1]$. Thus, $k$ and $c$ determine the value ranges for $s_1 .. s_{d-1}$ and $s_d$, respectively. We recommend setting $k = c$ to ensure the consistency in the range of each token  within tabids. To configure $l$ and $k$, a larger $l$ allows finer differentiation of first-view embeddings during the clustering. Preliminary experiments suggest that $l=2$ is generally sufficient. The clustering tree is expected to have a depth of $2l$ to ensure balanced clustering for both views, resulting in a maximum tabid length $d=2l+1$. In a perfect $k$-ary clustering tree, where each leaf node contains $k$ embedding vectors, the maximum number of tables $N$ it can represent is $k^{2l+1}$ if both views consist of $l$ levels, or $k^{l+1}$ if only the first view is considered. To ensure balanced clustering across two views, we recommend setting $k$ within the range $\left(N^{\frac{1}{2l+1}}, N^{\frac{1}{l+1}}\right)$. Additionally, we perform a sensitivity study to explore impact of varying $k$ within this range.
As shown in Table~\ref{tab:vary_k}, search accuracy remains stable, demonstrating \textsc{Birdie}'s robust performance.

\begin{table}[H]
\small
\centering
\caption{The accuracy of \textsc{Birdie} under different $k$ within the recommended range.}
\label{tab:vary_k}
\renewcommand{\arraystretch}{1.0} 
\vspace{-0.1in}
\setlength{\tabcolsep}{2.2mm}{
\begin{tabular}{ccc|ccc|ccc} 
\specialrule{.12em}{.06em}{.06em}
\multicolumn{3}{c|}{NQ-Tables} & \multicolumn{3}{c|}{FetaQA} & \multicolumn{3}{c}{OpenWikiTable}  \\ 
\hline
$k$ & $P@1$ & $P@5$ &            $k$ & $P@1$ & $P@5$ &                  $k$ & $P@1$ & $P@5$                            \\ 
\hline
16  &  45.10 &  70.23            &8   & 84.77  & 90.86                     &12   & 95.71  &  97.88                            \\
32  &  46.64 &  73.21            &14  & 86.32 &  92.76                    &20    & 97.20  & 99.06                              \\
48  &  45.06 &  71.22                    &20  & 86.27 & 92.56           &28 & 96.32 & 98.11                             \\
\specialrule{.12em}{.06em}{.06em}
\end{tabular}}
\end{table}

\begin{figure}[h]
\centering
\vspace{-2mm}
\subfigure[$P@1$ on NQ-Tables]{
    \includegraphics[width=0.47\linewidth]{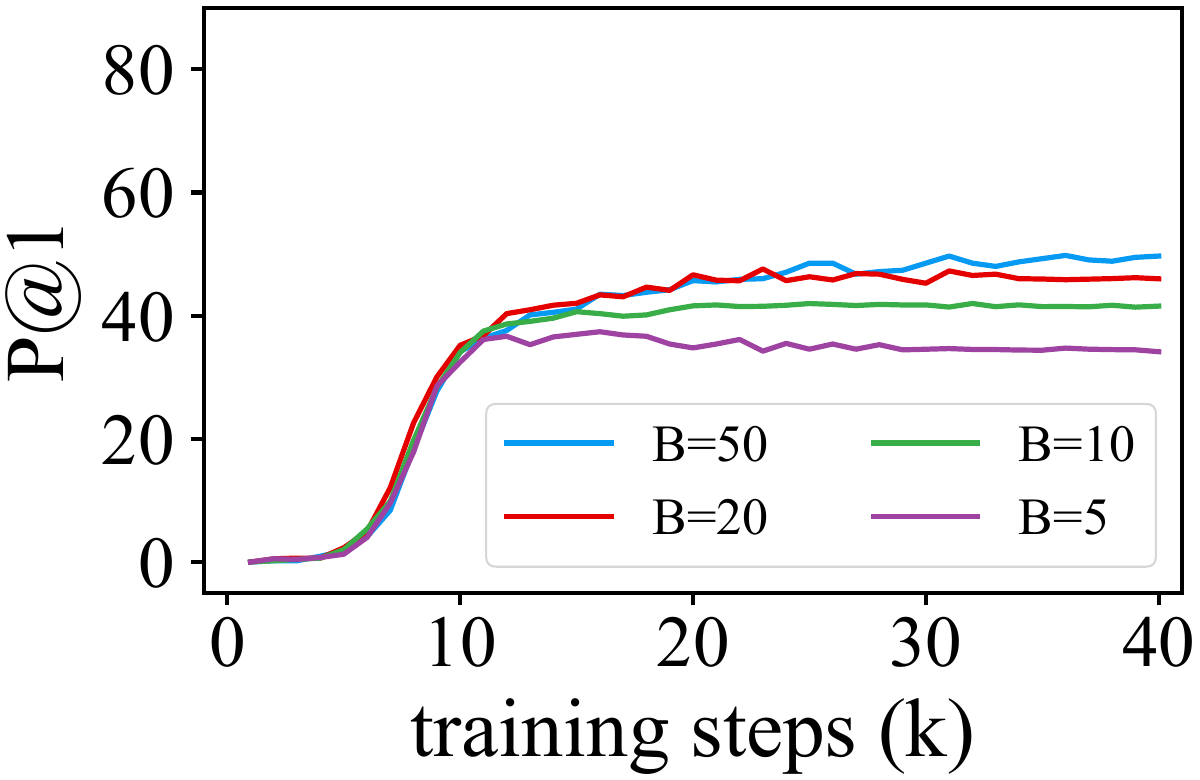}
}
\subfigure[$P@5$ on NQ-Tables]{
\includegraphics[width=0.47\linewidth]{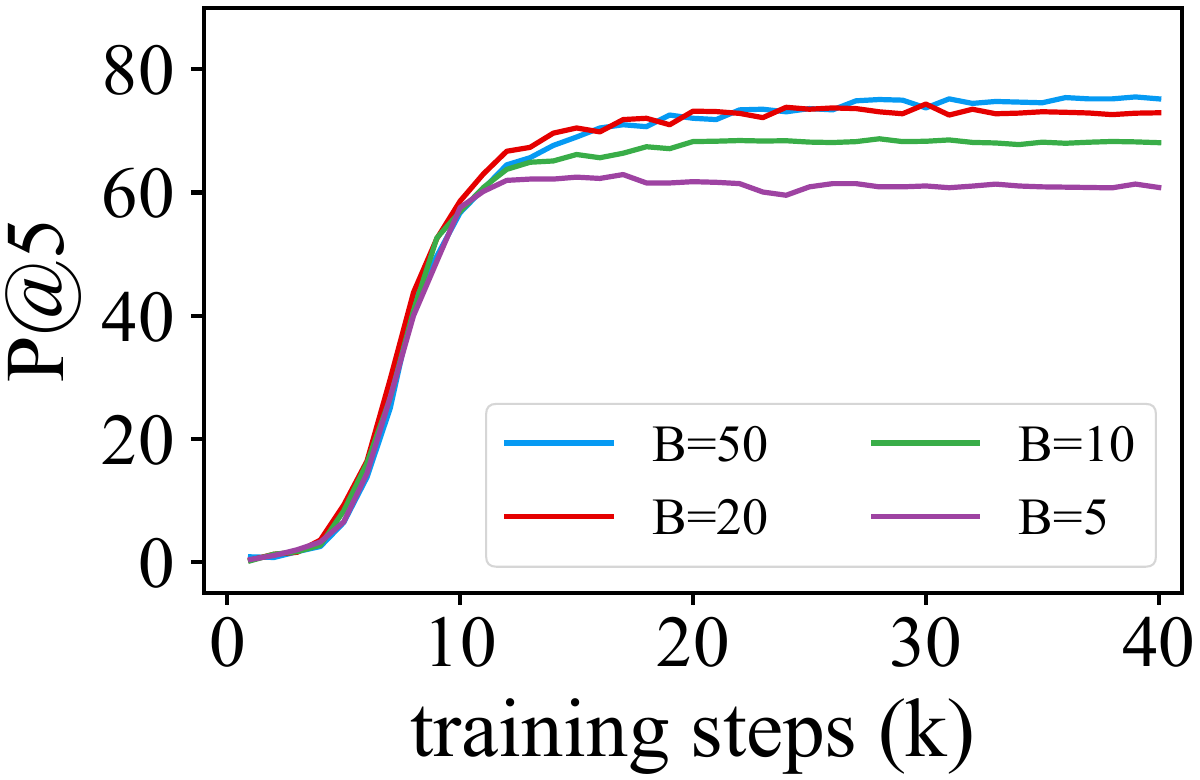}
}\\
\vspace{-0.03in}
\subfigure[$P@1$ on FetaQA]{
    \includegraphics[width=0.47\linewidth]{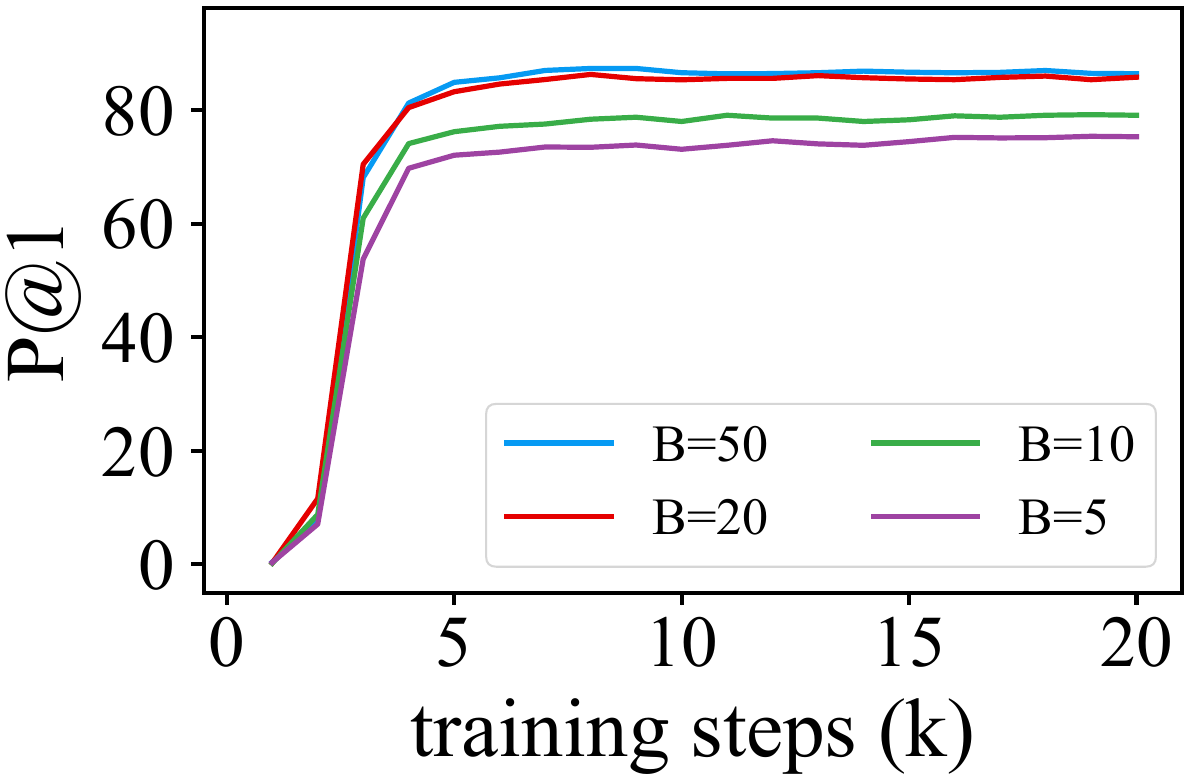}
}
\subfigure[$P@5$ on FetaQA]{
    \includegraphics[width=0.47\linewidth]{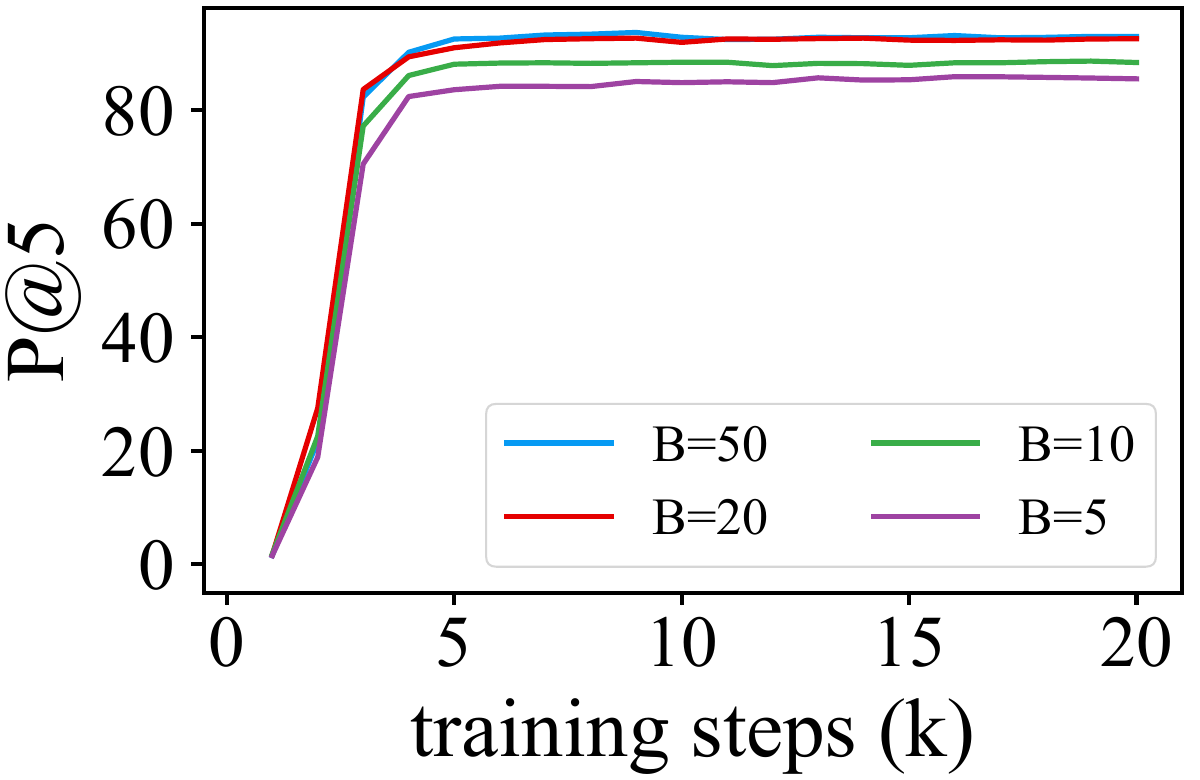}
}\\
\vspace{-0.03in}
\subfigure[$P@1$ on OpenWikiTable]{
    \includegraphics[width=0.47\linewidth]{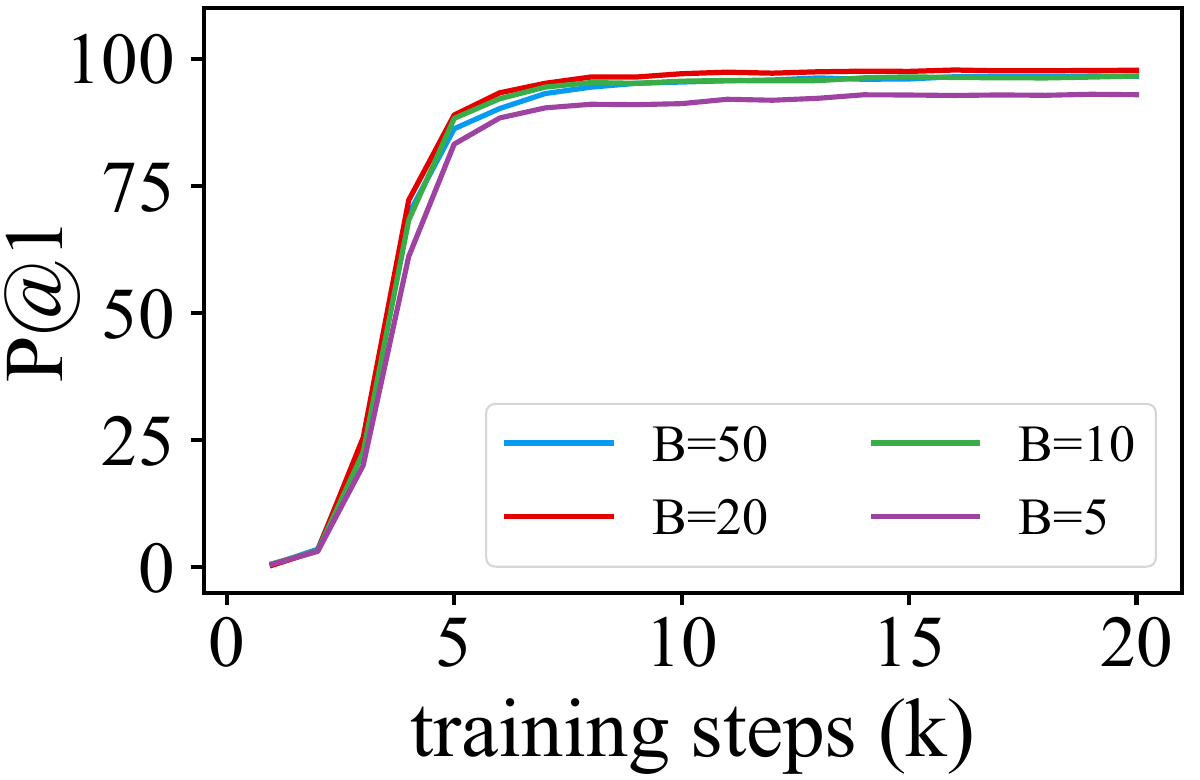}
}
\subfigure[$P@5$ on OpenWikiTable]{
    \includegraphics[width=0.47\linewidth]{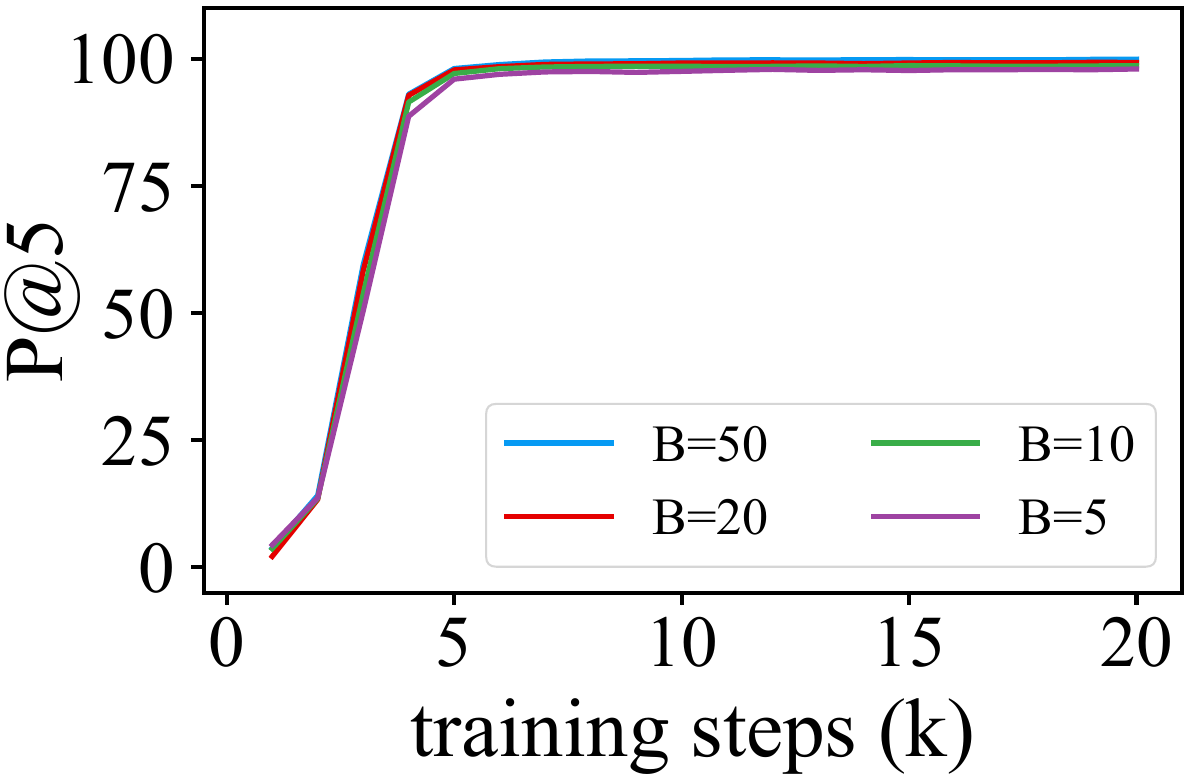}
}\\

\caption{Performance with different \# of generated queries.}
\label{fig:vary_q}

\end{figure}

\section{Supplementary Experimental Results}

\subsection{Runtime of index update}
\label{appendix:C.1}
Figure~\ref{fig:inc_time_appen1} and Figure~\ref{fig:inc_time_appen2}  illustrate  the time taken for index update across sequential batches $D^0$ to $D^3$ on FetaQA and OpenWikiTable. The observation is consistent across different datasets. Specifically, for tabid assignment, \textsc{Birdie}'s running time is 1-2 orders of magnitude shorter than that of \textsf{ReIndex}. This is because \textsf{ReIndex} needs to assign ids for all the tables from scratch, while \textsc{Birdie} only assigns tabids to new tables without affecting the ids of old  tables. For training time, \textsf{DSI++} is the fastest (the runtime of \textsf{CLEVER} is similar and omitted), utilizing less data than \textsf{Full} and \textsf{ReIndex} and fine-tuning the model from the last checkpoint, though its effectiveness is limited. \textsc{Birdie} demonstrates greater efficiency than \textsf{Full} and \textsf{ReIndex}, while achieving comparable average performance.

\begin{figure}[t]
\centering
\quad \quad \includegraphics[width=0.6\linewidth]{exp_fig/legend_time1.pdf}\\
\vspace{-2mm}

\subfigure[Tabid assignment time]{\begin{minipage}[t]{0.42\linewidth}
\includegraphics[width=1\linewidth]{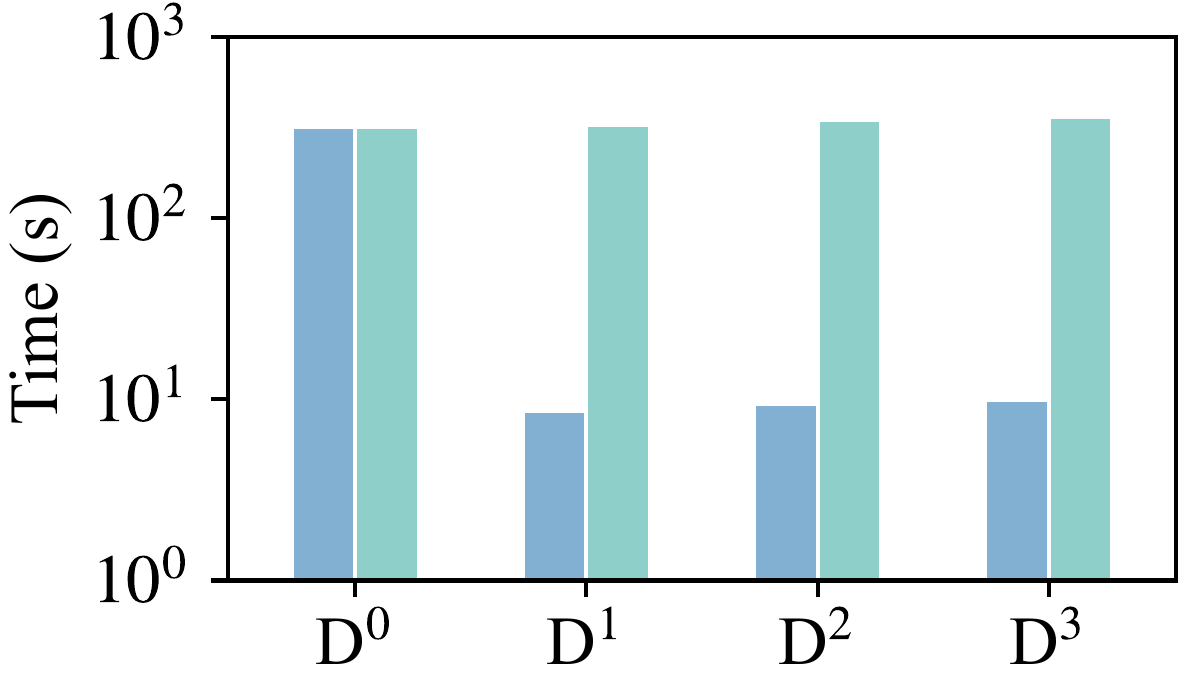}
\end{minipage}
}
\subfigure[Training time]{\begin{minipage}[t]{0.42\linewidth}
    \includegraphics[width=1\linewidth]{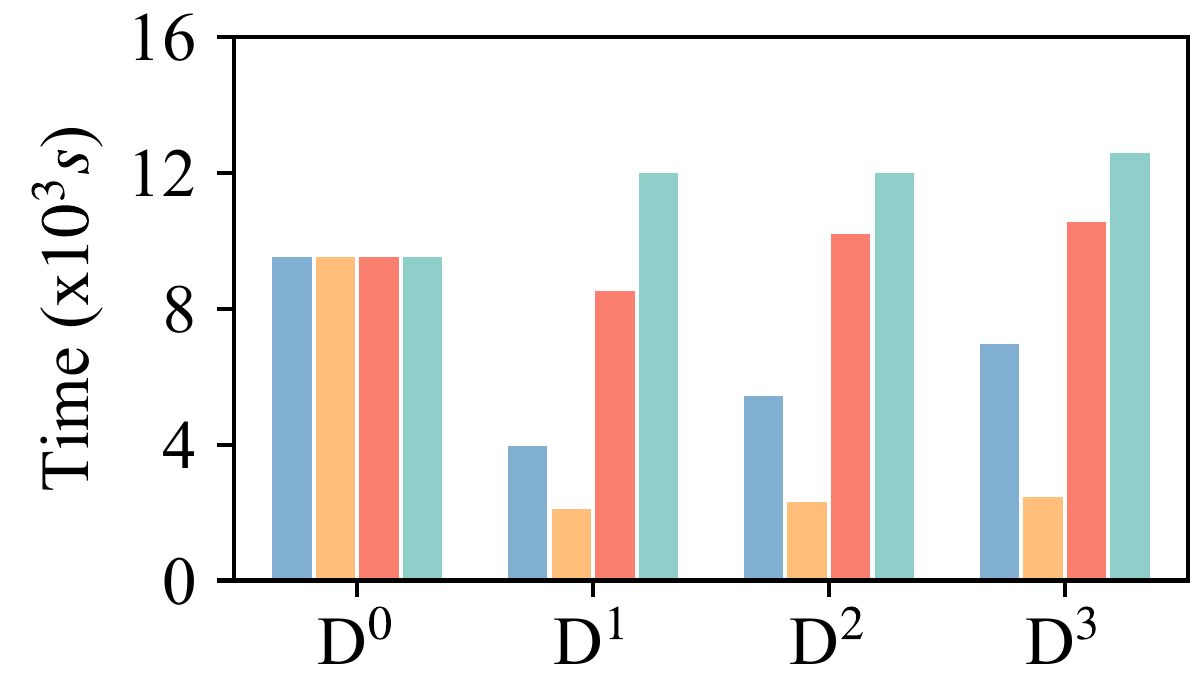}
\end{minipage}
}

\vspace{-4mm}
\caption{Runtime of index update on FetaQA.}
\label{fig:inc_time_appen1}
\vspace{-3mm}
\end{figure}

\begin{figure}[t]
\centering
\quad \quad \includegraphics[width=0.6\linewidth]{exp_fig/legend_time1.pdf}\\
\vspace{-2mm}

\subfigure[Tabid assignment time]{\begin{minipage}[t]{0.42\linewidth}
\includegraphics[width=1\linewidth]{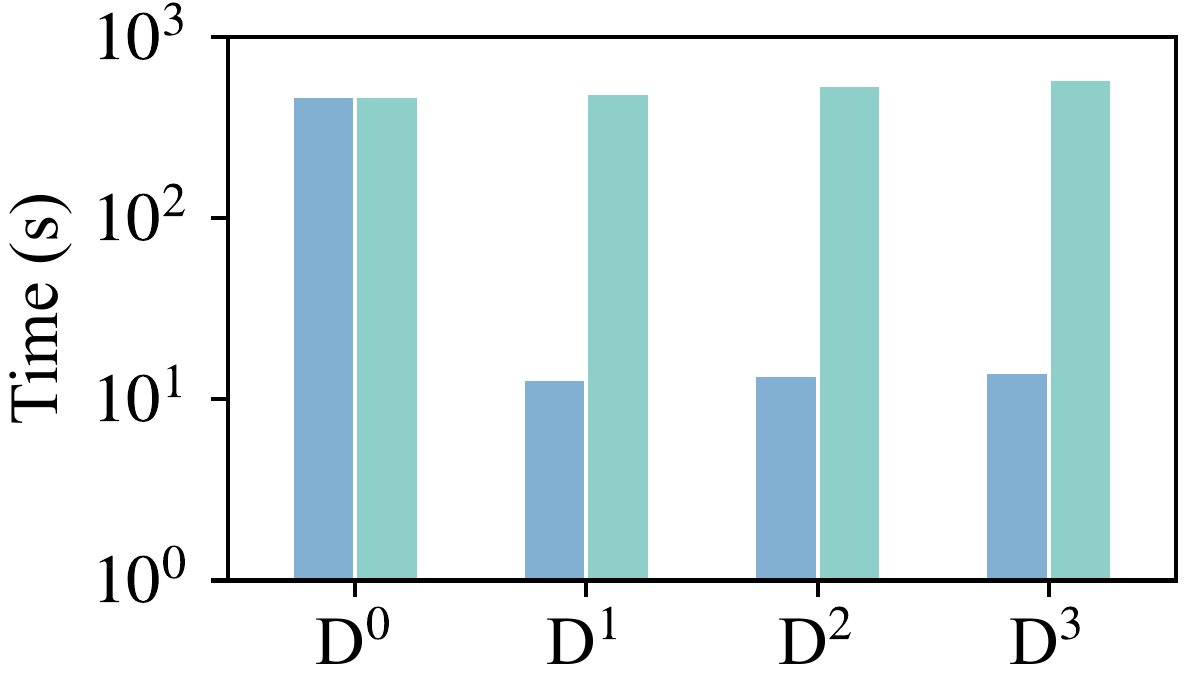}
\end{minipage}
}
\subfigure[Training time]{\begin{minipage}[t]{0.42\linewidth}
    \includegraphics[width=1\linewidth]{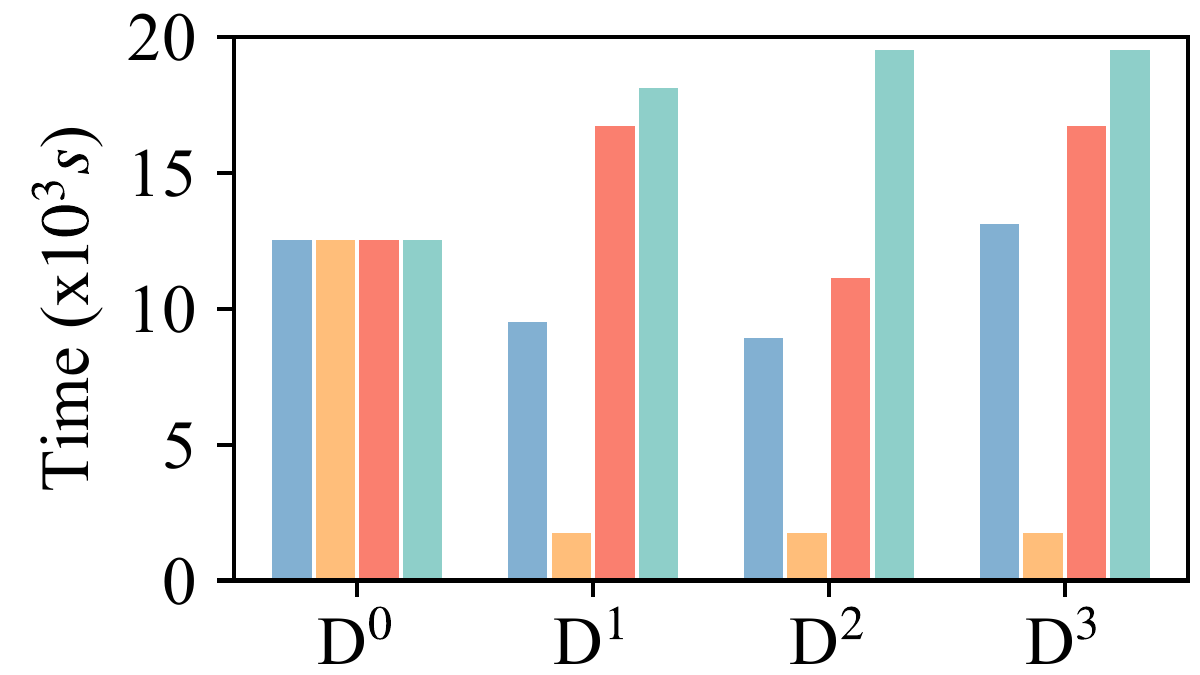}
\end{minipage}
}
\vspace{-4mm}
\caption{Runtime of index update on OpenWikiTable.}
\label{fig:inc_time_appen2}
\vspace{-2mm}
\end{figure}

\subsection{Effectiveness in the absence of old retrieval}
\label{appendix:C.2}
In the extreme case where old memory has expired, \textsc{Birdie} can easily adapt by adjusting the selection probability of candidate tables $\{\mathcal{T}^0, \dots, \mathcal{T}^u\}$ to improve its performance, where $\mathcal{T}^i = \{T^i_1, \dots, T^i_K\} \subset D^i$ represents tables generated by model $\mathcal{M}^i$ given an input query. Specifically, when old memory has been rarely retrieved over a period, \textsc{Birdie} can increase the probability of selecting $\mathcal{T}^{w} \subset D^w$ as the final result, where $D^w$ is a new batch. 
To demonstrate the effectiveness of this adjustment, we have implemented a straightforward strategy. After query mapping, we define $\mathcal{Q}^{z}$ as the set with the largest number of queries among the top-$n_q$ similar queries, and $\mathcal{Q}^{w}$ as the second largest. We prioritize $\mathcal{T}^{w}$ over $\mathcal{T}^{z}$ if $D^{w}$ is a new batch and $D^z$ is an old sub-repository. The results, presented in Table~\ref{tab:LP@1} and Table~\ref{tab:LP@5}, demonstrate that this simple extension allows \textsc{Birdie} to outperform baselines in most cases, highlighting its robustness in extreme scenarios.

\begin{table}[t]
\small
\centering
\caption{ $LP@1$ results of \textsc{Birdie} with new memory prioritization and baseline methods.} 
\renewcommand{\arraystretch}{1.1}  
\vspace{-0.1in}
\label{tab:LP@1}
\setlength{\tabcolsep}{0.8mm}{
\begin{tabular}{c|ccc|ccc|ccc} 
\specialrule{.12em}{.06em}{.06em}
\multirow{2}{*}{Methods} & \multicolumn{3}{c|}{NQ-Tables} & \multicolumn{3}{c|}{FetaQA} & \multicolumn{3}{c}{OpenWikiTable}  \\ 
\cline{2-10}
                        & $D^1$ & $D^2$ & $D^3$                     & $D^1$ & $D^2$ & $D^3$                  & $D^1$ & $D^2$ & $D^3$                          \\ 
\hline

\textsf{DSI++}                  &45.26   &40.74   &39.06                        & 88.08  & \textbf{90.23}  & 88.65                    & 95.58  &93.20   & 92.73                           \\
\textsf{CLEVER}                   & 48.42   & 43.60   & 41.45                     & 86.53 & 89.24  & 87.83                      &94.16   &92.48   & 92.2                           \\
\textsc{Birdie}                   &\textbf{53.68}   & \textbf{49.69}  &  \textbf{47.41}                     &\textbf{89.12}   & 89.48   & \textbf{88.98}                   & \textbf{98.43}  & \textbf{97.69}   &\textbf{97.36}                            \\
\specialrule{.12em}{.06em}{.06em}

\end{tabular}}
\end{table}

\begin{table}[t]
\small
\centering

\caption{ $LP@5$ results of \textsc{Birdie} with new memory prioritization and baseline methods.} 
\renewcommand{\arraystretch}{1.1}
\vspace{-0.1in}
\label{tab:LP@5}
\setlength{\tabcolsep}{0.8mm}{
\begin{tabular}{c|ccc|ccc|ccc} 
\specialrule{.12em}{.06em}{.06em}
\multirow{2}{*}{Methods} & \multicolumn{3}{c|}{NQ-Tables} & \multicolumn{3}{c|}{FetaQA} & \multicolumn{3}{c}{OpenWikiTable}  \\ 
\cline{2-10}
                         & $D^1$ & $D^2$ & $D^3$                      & $D^1$ & $D^2$ & $D^3$                  & $D^1$ & $D^2$ & $D^3$                         \\ 
\hline

\textsf{DSI++}                  &74.73   &72.71   & 69.90                       & 93.78  & \textbf{94.34}  & 92.90                    &99.28   &99.05   & 98.99                           \\
\textsf{CLEVER}                   & 75.79  &74.53   &72.07                        & 92.75   & 94.05   & 93.36                    & 99.00  &98.63   & 98.04                           \\
\textsc{Birdie}                   &\textbf{84.21}   &\textbf{80.90}   &\textbf{77.74}                        &\textbf{95.34}   & 93.86  & \textbf{93.40}                    & \textbf{99.90}  & \textbf{99.34}   & \textbf{98.87}                           \\
\specialrule{.12em}{.06em}{.06em}
\end{tabular}}
\end{table}

\begin{table*}[h]
\small
\centering
\caption{Examples of query and its corresponding table's title information in OpenWikiTable  }
\label{tab:case_openwiki}
\vspace{-2mm}
\setlength{\tabcolsep}{0.7mm}{
\begin{tabular}{|l|l|} 
\hline
\multicolumn{1}{|c|}{\textbf{Query}} & \multicolumn{1}{c|}{\textbf{Title of the corresponding table}}  \\ 
\hline
 \makecell[c]{What is the pinnacle height for Metcalf, Georgia among the \\  \colorbox{green!15}{list of tallest structures in the world}  \colorbox{green!15}{(past or present) 600 m and taller (1,969 ft)} ?}  &  \makecell[l]{\texttt{Page Title}: \colorbox{green!15}{List of tallest structures in the world} \\ \texttt{Caption}: Structures \colorbox{green!15} {(past or present) 600 m and taller (1,969 ft)}}             \\
\hline
  \makecell[c]{From the  \colorbox{green!15}{2009 Supersport World Championship season}, if the \\ Autodromo Nazionale Monza is the circuit, what's the report?}    &   \makecell[l]{\texttt{Page Title}: \colorbox{green!15}{2009 Supersport World Championship season} \\   \texttt{Caption}: Race calendar and results }\\ 
\hline
  \makecell[c]{What was the outcome of the district represented by Thomas R. Gold \\ in the  \colorbox{green!15}{United States House of Representatives elections of 1812} ?}                          &   \makecell[l]{\texttt{Page Title}: \colorbox{green!15}{United States House of Representatives elections, 1812} \\   \texttt{Caption}: New York }                           \\ 
\hline
\makecell[c]{What is the type of electronics with the Gamecube platform under \\ \colorbox{green!15}{Harry Potter and the Chamber of Secrets} ?} & \makecell[l]{\texttt{Page Title}: \colorbox{green!15}{Harry Potter and the Chamber of Secrets} \\ \texttt{Caption}: Video game}\\
\hline
\end{tabular}}
\vspace{3mm}
\end{table*}

\section{Supplementary Case Study}
\subsection{Examples of query-table pairs from OpenWikiTable}
All methods perform better on OpenWikiTable compared to the other two datasets. This can be attributed to: (i) the queries in OpenWikiTable benchmark contain detailed descriptions of the ground truth table (i.e. page title, section title, caption), making it easier to locate the relevant table using textual similarity, even using a dense passage retriever; and (ii) a remarkable 99.8\% (24,634 out of 24,680) of tables in OpenWikiTable have unique titles, simplifying the task of identifying the correct table using title information alone. While reasoning across multiple rows and columns in OpenWikiTable is challenging for QA tasks, this does not make the table discovery process equivalently difficult. Instead, table discovery is relatively straightforward for OpenWikiTable due to the substantial overlap between the queries and the titles of ground truth tables. 
Table~\ref{tab:case_openwiki} presents examples of query-table pairs from OpenWikiTable, highlighting why table discovery in this dataset is simpler.

\begin{figure*}[t]
\vspace{2mm}
  \centering
  \includegraphics[width=1\linewidth]{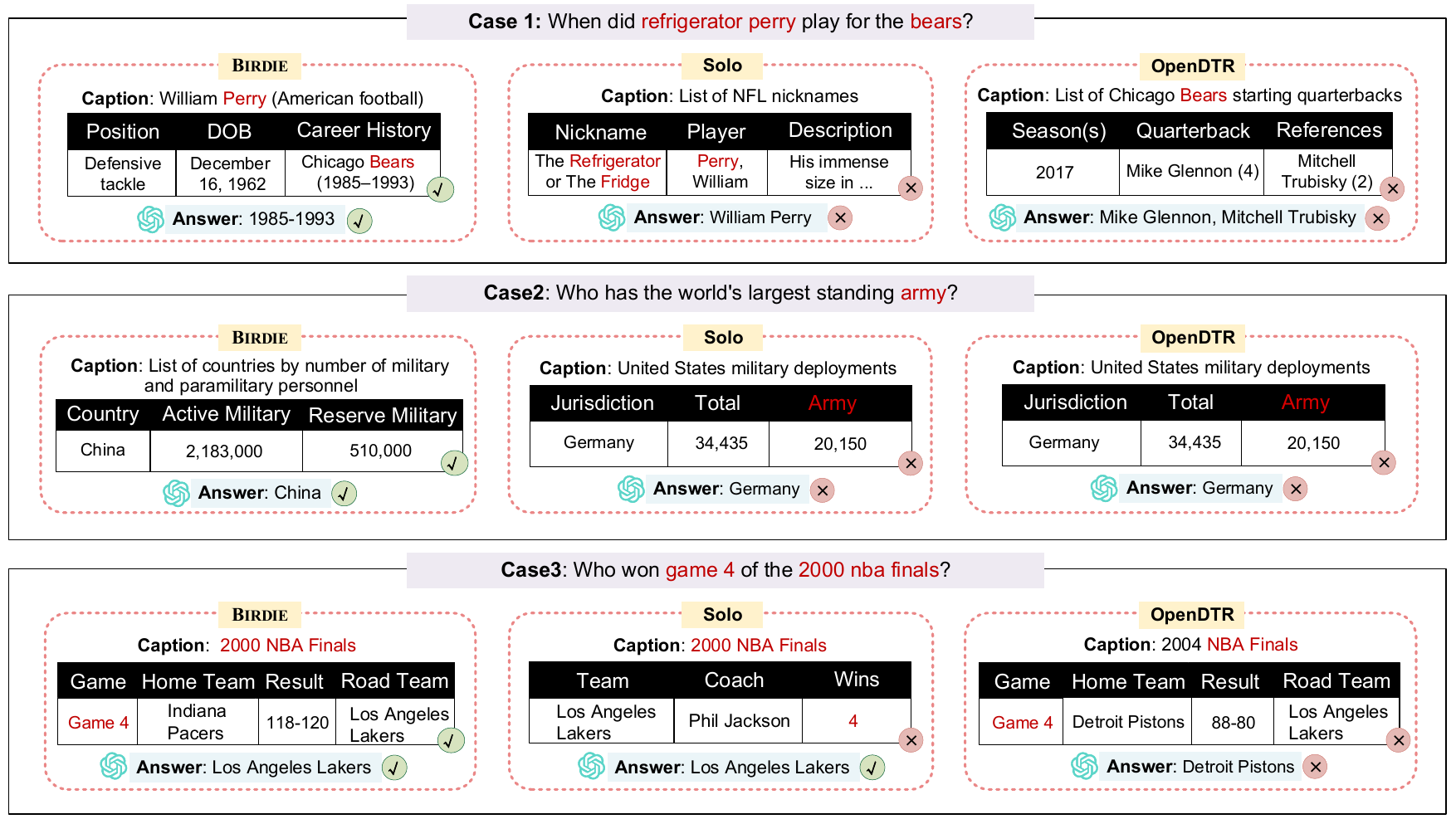}\vspace{-2mm}
  \caption{{NL-driven table discovery and QA results using GPT-4o, showing only the row serving as evidence for the QA result.}
  \label{fig:case_appendix}}
    \vspace{-3mm}
\end{figure*}

\subsection{Cases of table discovery and table QA}

We include additional two representative cases  to highlight the advantages of $\textsc{Birdie}$, illustrated in Figure~\ref{fig:case_appendix}. In Case 2, $\textsc{Birdie}$ successfully captures the key intent of the query rather than blindly matching overlapping terms (e.g., ``army") between the query and table attributes. In Case 3, despite the presence of a specific sports term (e.g., ``NBA Finals'') in the query, identifying the correct table remains challenging, as multiple tables contain identical or similar titles with this jargon. 
Nevertheless, only \textsc{Birdie} identifies the correct table. 
While the table QA result from the table retrieved by \textsf{Solo} is correct, it is coincidental.
Specifically, \textsf{Solo} retrieves a table showing the total number of wins for each team in the 2000 NBA Finals but fails to specify the winner of individual games, such as Game 4. The correct answer is coincidental because the team with 4 total wins happened to win Game 4.
\end{document}